\documentclass{article}

\usepackage[margin=1in]{geometry}
\usepackage[utf8]{inputenc} %
\usepackage[T1]{fontenc} %
\usepackage{libertine}
\usepackage[scr=rsfs]{mathalpha}
\usepackage{multirow}

\usepackage{graphicx}

\usepackage{authblk} %

\usepackage{caption} %
\usepackage{subcaption} %
\usepackage{csquotes}

\usepackage{microtype}
\microtypesetup{
    activate={true,nocompatibility}, %
    disable=false, %
    kerning=true, %
    spacing=true,
    factor=1100, %
    stretch=10, %
    shrink=10 %
}
\microtypecontext{spacing=nonfrench} %

\usepackage[normalem]{ulem}
\usepackage[dvipsnames]{xcolor}
\usepackage{amsmath}
\usepackage{amsthm}
\usepackage{thm-restate}
\usepackage{amsfonts}
\usepackage{dsfont}
\usepackage{algorithm}
\makeatletter
\renewcommand{\ALG@name}{Mechanism}
\makeatother
\usepackage[italicComments=true,commentColor=teal]{algpseudocodex}
\newtheorem{theorem}{Theorem}
\newtheorem{lemma}{Lemma}[section]
\newtheorem{corollary}[lemma]{Corollary}
\theoremstyle{definition}
\newtheorem{definition}{Definition}[section]
\newtheorem{claim}[lemma]{Claim}
\usepackage{hyperref}
\hypersetup{
    colorlinks=true,
    linkcolor=red, %
    citecolor=blue, %
    filecolor=cyan, %
    urlcolor=purple, %
}

\usepackage{tcolorbox}

\usepackage{marginnote}
\newcommand{\info}[1]{}

\newcounter{game}


\newcommand\eps\epsilon

\newtheorem{fact}{Fact}

\newcommand{\universe}{\ensuremath{\chi}}
\newcommand{\cmax}{\ensuremath{q^*}}

\newcommand{\maxsum}{\textsc{MaxSum}}
\newcommand{\minsum}{\textsc{MinSum}}
\newcommand{\topk}{\textsc{TopK}}

\newcommand{\sumselect}{\textsc{SumSelect}}
\newcommand{\histogram}{\textsc{Histogram}}
\newcommand{\median}{\textsc{Median}}
\newcommand{\quantile}[1]{\textsc{Quantile}_{#1}}

\newcommand{\alg}[1]{\mathcal{A}(#1)}

\newcommand{\Lap}{\mathrm{Lap}}

\newcommand{\densLap}[1]{f_{\Lap(#1)}}

\newcommand{\diff}{\mathrm{diff}}
\newcommand{\mode}{\mathit{mode}}

\newcommand{\Thresh}{\mathrm{Thresh}}
\newcommand{\err}{\mathrm{err}}

\newcommand{\histshort}{\mathrm{h}}
\newcommand{\out}{\mathrm{out}}
\newcommand{\thresh}{\Thresh}
\newcommand{\side}{\mathrm{side}}
\newcommand{\type}{\mathrm{type}}
\newcommand{\T}{\thresh}

\newcommand{\tO}{\ensuremath{\widetilde{O}}}

\newcommand{\amu}[1]{\ensuremath{\alpha_{\mu}^{#1}}}
\newcommand{\atau}[1]{\ensuremath{\alpha_{\tau}^{#1}}}
\newcommand{\agamma}[1]{\ensuremath{\alpha_{\gamma}^{#1}}}
\newcommand{\ah}[1]{\ensuremath{\alpha_{H}^{#1}}}
\newcommand{\amut}{\amu{t}}
\newcommand{\atauj}{\atau{j}}
\newcommand{\agammaj}{\agamma{j}}
\newcommand{\ahj}{\ah{j}}
\newcommand{\threesum}[2]{\amu{#1} + \atau{#2} + \agamma{#2}}
\newcommand{\threesumtj}{\threesum{t}{j}}
\newcommand{\foursum}[2]{\amu{#1} + \atau{#2} + \agamma{#2} + \ah{#2}}
\newcommand{\foursumj}{\foursum{t}{j}}

\newcommand{\amutval}{\ensuremath{12\epsilon^{-1} \ln(2/\beta_t)}}
\newcommand{\ataujval}{\ensuremath{6\epsilon^{-1} \ln(6/\beta_j)}}
\newcommand{\agammajval}{\ensuremath{3\epsilon^{-1} m \ln (6m/\beta_j)}}
\newcommand{\ahjval}{\ensuremath{\errsix{j}}}
\newcommand{\murv}{$\Lap(12/\epsilon)$}
\newcommand{\taurv}{$\Lap(6/\epsilon)$}
\newcommand{\gammarv}{$\Lap(3m/\epsilon)$}

\newcommand{\kcmax}{\ensuremath{{m \cmax}}}
\newcommand{\kcmaxt}{\ensuremath{m { \cmax }^t}}

\newcommand{\edee}{\ensuremath{(\epsilon/3,\delta/(2e^{2\epsilon/3}))}}

\newcommand{\edu}{\ensuremath{6\epsilon^{-1} \sqrt{m \ln (12 e^{2\epsilon/3} m /(\delta\beta_j))} }}
\newcommand{\edgammarv}{\ensuremath{N(0, 18k\ln (4e^{2\epsilon/3}/\delta)/\epsilon^2)}}

\newcommand{\errsix}[1]{\ensuremath{\errgen{#1, \beta_{#1}/6}}}

\newcommand{\jinterval}{\ensuremath{[p_{j-1}, p_j)}}

\newcommand{\he}{\ensuremath{\epsilon/3}}
\newcommand{\hd}{\ensuremath{\delta/(2e^{2\epsilon/3})}}

\newcommand{\assumeconc}{Assume \bounds\ holds. }

\newcommand{\bounds}{Lemma~\ref{applem:accconc}}
\newcommand{\lb}{Lemma~\ref{applem:accgiLB}}
\newcommand{\ub}{Lemma~\ref{applem:accgiUB}}
\newcommand{\lapbound}{Lemma~\ref{fact:laplace_tailbound}}

\newcommand{\pj}{\ensuremath{{p_{j}}}}
\newcommand{\last}{\ensuremath{{\ell}}}
\newcommand{\plast}{\ensuremath{{p_{\last}}}}

\newcommand{\tfirst}{\ensuremath{{t_{\mathrm{first}}}}}

\newcommand{\eaccvalue}{\ensuremath{{ O\left(\left( d \log^2(d\kcmax/\beta) + m\log(\kcmax/\beta) + \log t \right) \epsilon^{-1} \right) }}}

\newcommand{\accalg}{Mechanism~\ref{alg:histquery}}

\newcommand{\Ki}{\ensuremath{L_{i}}}
\newcommand{\Kit}{\ensuremath{{\Ki^t}}}
\newcommand{\Kitplusone}{\ensuremath{{\Ki^{t+1}}}}
\newcommand{\Kipj}{\ensuremath{{\Ki^\pj}}}
\newcommand{\Kiplast}{\ensuremath{{\Ki^\plast}}}

\newcommand{\Kione}{\ensuremath{{\Ki^1}}}

\newcommand{\errgen}[1]{\mathrm{err}(#1)}
\newcommand{\threshcross}[1]{#1 crosses the threshold}
\newcommand{\notthreshcross}[1]{#1 did not cross the threshold}

\newcommand{\unbddaccj}{\ensuremath{O(\epsilon^{-1} d \cdot ( \log (j) \log(d/\beta) + (\log j)^{1.5} \sqrt{\log(d/\beta)} ))}}

\newcommand{\alphatime}[1]{\ensuremath{\alpha^{(#1)}}}
\newcommand{\alphat}{\ensuremath{\alphatime{t}}}

\newcommand{\yes}{\textsc{Yes}}
\newcommand{\no}{\textsc{No}}
 
\usepackage{natbib}
\begin{document}

\title{Differentially Private Continual Release \\ of Histograms and Related Queries}

\author[1]{Monika Henzinger}
\affil[1]{Institute of Science and Technology Austria (ISTA), Klosterneuburg, Austria}
\author[2]{A. R. Sricharan}
\affil[2]{Faculty of Computer Science, University of Vienna, Austria}
\affil[2]{UniVie Doctoral School Computer Science DoCS}
\author[3]{Teresa Anna Steiner}
\affil[3]{University of Southern Denmark Odense, Denmark}

\maketitle

\begin{abstract}
We study privately releasing column sums of a $d$-dimensional table with entries from a universe $\chi$ undergoing $T$ row updates, called histogram under continual release.
Our mechanisms give better additive $\ell_\infty$-error than existing mechanisms for a large class of queries and input streams.
Our first contribution is an output-sensitive mechanism in the insertions-only model ($\chi = \{0,1\}$) for maintaining (i) the histogram or (ii) queries that do not require maintaining the entire histogram, such as the maximum or minimum column sum, the median, or any quantiles.
The mechanism has an additive error of $O(d\log^2 (d\cmax)+\log T)$ whp, where $\cmax$ is the maximum output value over all time steps on this dataset. The mechanism does not require $\cmax$ as input.
This breaks the $\Omega(d \log T)$ bound of prior work when $\cmax \ll T$.
Our second contribution is a mechanism for the turnstile model  that admits negative entry updates
($\chi = \{-1, 0,1\}$). This mechanism has an additive error of $O(d \log^2 (dK) + \log T)$ whp, where $K$ is the number of times two consecutive data rows differ, and the mechanism does not require $K$ as input.
This is useful when monitoring inputs that only vary under unusual circumstances.
For $d=1$ this gives the first private mechanism with error $O(\log^2 K + \log T)$ for continual counting in the turnstile model, improving on the $O(\log^2 n + \log T)$ error bound by \citet{conf/asiacrypt/DworkNRR15}, where $n$ is the number of ones in the stream,
as well as allowing negative entries, while \citet{conf/asiacrypt/DworkNRR15}
can only handle nonnegative entries ($\chi=\{0,1\}$).
 \end{abstract}

\section{Introduction}

\info{Motivate \\ continual \\ counting}
Maintaining continual sums of a stream of numbers is an integral subroutine in various applications, including iterative first-order methods in machine learning and online convex optimization, which require maintaining the sum of all the past gradients to make future decisions. Private mechanisms for these problems thus rely on privately maintaining continual sums of gradients~\citep{kairouz2021practical, zhang22onlinetobatch, neurips2022, choo23factorization,asi2023near, choo23banded, xu23gboard}, or more generally of a stream of numbers~\citep{Dwork2010, DBLP:journals/tissec/ChanSS11, DBLP:conf/icml/FichtenbergerHU23, henzinger23counting, Andersson23smoothbm,DBLP:conf/soda/HenzingerU25}, with minimal error. The privacy model used
is the well-studied
model of \emph{differential privacy under continual observation}~\citep{Dwork2006, Dwork2010}, which requires that the output distributions of a private mechanism are close on neighboring streams. When the elements of the stream are nonnegative, this is the \emph{insertions-only} model, while the \emph{turnstile model} allows negative numbers as well.

\info{DP CC \\ history \\ and \\ motivation}
Since the initial work of \citet{Dwork2010} and \citet{DBLP:journals/tissec/ChanSS11},
achieving an asymptotic improvement in the additive error of private continual counting or proving that the current bound is optimal has become a major open problem in the field.
Recent work has concentrated on non-asymptotic improvements:
(1)  \citet{DBLP:conf/icml/FichtenbergerHU23,andersson2024countelderslaplacevs,DBLP:conf/soda/HenzingerU25} improved the constant term in front of the larger asymptotic terms in the additive error.
(2) \citet{conf/asiacrypt/DworkNRR15} improved the error to optimal on sparse streams in the insertions-only setting with an error bound that is parameterized in the maximum output value.
Currently, all improvements to the additive error are interesting since the huge error terms in existing mechanisms lead to large choices of privacy parameters in practice to ensure accurate outputs, which  in turn leads to reduced privacy guarantees for these implementations.
We present two improvements to a high-dimensional generalization of this problem in the spirit of the latter parameterized result of \citet{conf/asiacrypt/DworkNRR15}.

At a high level, \citet{conf/asiacrypt/DworkNRR15} obtain their parameterized improvements by partitioning the stream of inputs into \emph{intervals}, then batching all the inputs within an interval as a single input to a black-box continual counting mechanism. They show that the number of intervals created is proportional to the sparsity of the input stream. This reduces the error dependence on the length of the input stream to just its sparsity.

\info{DP Hist}
Extending continual sum of a stream of scalars to maintaining continual coordinate sums of a stream of $d$-dimensional vectors is a natural generalization when one needs to keep track of high-dimensional information,
as is the case in the above machine learning applications.
We study this setting, called the \histogram\ problem, which requires maintaining for every $i \in [d]$, the sum of the $i$-th coordinate of all (row) vectors seen so far, called the $i$-th \emph{column sum}.
We also show how to more accurately answer \emph{histogram queries} that do not necessarily require maintaining the entire histogram, for example, the minimum or %
median column sum.
Our error metric is the maximum error between the mechanism's output and the true output, where the maximum is over all coordinates, all time steps, and all queries, called the \emph{$\ell_{\infty}$-error}.
\begin{definition}[Continual Histogram]
\label{def:cont_hist}
Let $d \in \mathbb N_{> 0}$. The input is a stream of $T$ row vectors $x^1,\dots, x^T$ with $x^t\in\universe^d$ for all $t\in[T]$ where $\universe$ is the data universe and $T$ is not given as input.
In the \emph{insertions-only} setting, $\universe=\{0,1\}$. In the \emph{turnstile model}, $\universe=\{-1,0,1\}$.
The output of \histogram\ at each time step $t$ is (an additive approximation to) the column sums of all inputs seen so far, i.e., $\histshort(t)=(\sum_{\ell=1}^t x_i^\ell)_{i\in[d]}$.
When $d=1$, this is the \emph{continual counting} problem.
\end{definition}

An important definition in differential privacy is that of \emph{neighboring streams}, which specifies the level of privacy guaranteed by the mechanism. We use the standard event-level neighboring definition which allows two neighboring streams to differ in all the coordinates of exactly one of the $d$-dimensional input vectors.

\paragraph*{Motivating Questions}

\info{Motivation for \\ Mechanism 1}
We discuss $\eps$-differential privacy below, and drop $\eps^{-1}$ and $\log(1/\beta)$ factors ($\beta$ is the failure probability) in the discussion for simplicity.
For private continual counting, the best known lower bound on the additive error is $\Omega(\log T)$~\citep{Dwork2010} against an upper bound of $O(\log^2 T)$~\citep{Dwork2010, DBLP:journals/tissec/ChanSS11}. In the parameterized setting, \citet{conf/asiacrypt/DworkNRR15} presented a mechanism with $O(\log^2 n + \log T)$ error for insertions-only streams, where $n$ is the largest output of the mechanism on the entire stream.
Since this mechanism asymptotically improves on all other mechanisms for streams with continual count $T^{o(1)}$ and is optimal when the continual count is $O(2^{\sqrt{\log T}})$, we ask the following generalization to $d$-dimensional streams.
\begin{tcolorbox}
\begin{center}
Can we exploit sparsity of outputs for private continual \histogram\ and histogram queries to obtain a smaller additive error?
\end{center}
\end{tcolorbox}

\info{Mention \\ histqueries}
Note that the requirement for histogram queries is stronger than that of \histogram, since the histogram query of, say, the minimum column sum has a much smaller output than that of the entire \histogram.

\info{Why new \\ mechanisms}
The best known lower bound here is $\Omega(d + \log T)$~\citep{pricedpco, Dwork2010}, and we show that existing mechanisms for \histogram\ have an additive error of $\Omega(d \log T)$. Thus, achieving $o(d \log T)$ error necessitates new approaches. Our first mechanism improves on existing mechanisms for streams and queries with maximum output value $T^{o(1)}$, and breaks the $\Omega( d \log T)$ barrier on streams with maximum output value $o(2^{\sqrt{\log T}})$, which answers the highlighted question in the affirmative.
The mechanism works by partitioning the input stream as in \citet{conf/asiacrypt/DworkNRR15}, while taking the interaction between all $d$ columns into account. This extension to the multi-dimensional setting involves overcoming multiple technical difficulties which we detail later.

\info{Motivation for\\ Mechanism 2}
Another shortcoming of the mechanism of \citet{conf/asiacrypt/DworkNRR15} is that it does not allow negative entries. It implicitly assumes that the continual count is increasing \emph{monotonically}, which might not be true for turnstile streams. This motivates our second main question.
\begin{tcolorbox}
\begin{center}
Can we obtain parameterized improvements (similar to~\citet{conf/asiacrypt/DworkNRR15}) for turnstile streams?
\end{center}
\end{tcolorbox}

\info{Why turnstile \\ important}
Improvements in the turnstile model are important since continual histogram mechanisms in this model are used as subroutines in more advanced private continual release mechanisms, such as in the above machine learning applications, fully dynamic graph mechanisms~\citep{DBLP:conf/esa/FichtenbergerHO21,DBLP:conf/icml/FichtenbergerHU23}, for $k$-means clustering in Euclidean spaces~\citep{ICML24}, and also to count the \emph{difference sequence}~\citep{DBLP:conf/esa/FichtenbergerHO21} of an insertions-only stream.

\info{Our result}
We answer our second question in the affirmative, by parameterizing on the \emph{number of fluctuations}, $K$, which is the number of times
a row vector in the input is different from its immediately preceding row vector.
\info{Extensions}
We show that the improvements achieved by our first mechanism also carry over to the second, giving a mechanism for \histogram\ on turnstile streams that breaks the $\Omega(d \log T)$ barrier on streams with small $K$.
At a high level, while existing mechanisms use the sparse vector technique (SVT) to update outputs when the function value changes a lot, our mechanism instead uses SVT to check when the \emph{slope} of the function changes a lot, and we show that one can still obtain improved error bounds under this condition.

\paragraph*{Prior Work}
\label{sec:priorwork}

\info{Setup}
We discuss current approaches to \histogram\ and related queries for $\eps$-differential privacy, and present the $(\eps,\delta)$-dp bounds in the appendix.
The bounds discussed are also presented in Tables~\ref{tab:mech1} and \ref{tab:mech2}.

\info{existing \\ LBs}
On the lower bounds side, \citet{pricedpco} study \histogram\ as well as maintaining the maximum column sum (\maxsum) and its coordinate (\sumselect) privately and show that the additive error\footnote{We focus only on their bounds that are subpolynomial in $T$ since we consider the setting where $T \gg d$.\label{fn:subpolyT}} must be $\Omega(d)$. Combined with the counting lower bound of $\Omega(\log T)$ of \citet{Dwork2010}, this gives an $\Omega(d + \log T)$ bound. \citet{DBLP:conf/colt/Cohen0NSS24} give a lower bound of $\Omega(\min\{n, \log T\})$ for continual counting.
\info{existing \\ UBs}
On the upper bound side, \citet{pricedpco} give an $O(d\log d \log^3 T)$ bound on the additive error\footnote{Footnote \ref{fn:subpolyT} applies as well.}. In the turnstile model, composing $d$ binary tree counting mechanisms of \citet{DBLP:journals/tissec/ChanSS11}, and suitably scaling the privacy parameter and failure probability of each mechanism,
gives an error of $O(d \log^2 (dT))$.
For insertion-only streams, combining the mechanism of \citet{conf/asiacrypt/DworkNRR15} with observations by \citet{QiuYi} and \citet{DBLP:journals/tissec/ChanSS11} gives an error bound of $O({d\log^2 (dn_{\max}) + d\log T})$, where $n_{\max}$ is the maximum column sum.

\info{why new \\ mechanisms}
As seen from the bounds above, all existing histogram mechanisms have an $O(d \log T)$ term in their upper bounds. We show that this dependency is inherent, by proving that any histogram mechanism obtained as a composition of $d$ counting mechanisms has $\Omega(d \log T)$ error. This implies that any mechanism that beats this barrier must necessarily take into account the interaction between the different columns.

\begin{table*}[t]
\begin{center}
\caption{$\ell_{\infty}$-error of $\epsilon$-dp mechanisms for continual histogram queries with constant $\epsilon$, constant failure probability in the insertions-only model, where $n_{\max}$ is the maximum column sum, $\cmax$ is the maximum query output, and $K$ is the number of fluctuations in the input stream. $K$ and $\cmax$ do not need to be given to the mechanism. The binary tree result is by \citet{Dwork2010} and \citet{DBLP:journals/tissec/ChanSS11}, and the partitioning follows from \citet{conf/asiacrypt/DworkNRR15,DBLP:journals/tissec/ChanSS11} and \citet{QiuYi}.}
\label{tab:mech1}
\begin{tabular}{|c|c|c|}
\hline
Mechanism & \histogram  & $m$ histogram queries \\\hline
\cite{pricedpco}  & $O(d\log d\log^3 T)$  & $O(d\log d\log^3 T)$\\
\hline
Binary Tree & $O(d \log^2 (dT))$  & $O(d \log^2 (dT))$  \\\hline
Partitioning & $O({d\log^2 (dn_{\max}) + d\log T})$  & $O({d\log^2 (dn_{\max}) + d\log T})$\\\hline
\textcolor{purple}{Theorem \ref{thm:intro}}  & $O(d\log^2 (dn_{\max})+{\color{purple}\log T})$  & $O(d\log^2 (d{\color{purple}\cmax})+{\color{purple}m\log (m\cmax)}+{\color{purple}\log T})$ \\\hline
\textcolor{purple}{Theorem~\ref{thm:introNew}} & $O\left(d \log^2 ({\color{purple}dK}) +{\color{purple}\log T}\right)$  & $O\left(d \log^2 ({\color{purple}dK}) +{\color{purple}\log T}\right)$\\\hline
\end{tabular}
\end{center}
\end{table*}

\begin{table}
\begin{center}
\caption{$\ell_{\infty}$-error for $\epsilon$-dp mechanisms with constant $\epsilon$ and constant failure probability in the turnstile model, where $K$ is the number of fluctuations. The binary tree result is by \citet{Dwork2010} and \citet{DBLP:journals/tissec/ChanSS11}.}
\label{tab:mech2}
\begin{tabular}{|c|c|}
\hline
Mechanism & \histogram  \\\hline
\citet{pricedpco}  & $O(d\log d\log^3 T)$ \\
\hline
Binary Tree & $O(d\log^2 (dT))$  \\\hline
\textcolor{purple}{Theorem~\ref{thm:introNew}} &  $O\left(d \log^2 ({\color{purple}d K}) +{\color{purple}\log T}\right)$\\\hline
\end{tabular}
\end{center}
\end{table}

\section{Our Results}
\label{sec:our_results}

\info{our results}
To answer the questions raised above, we present two new mechanisms.
Both mechanisms achieve a parameterized error bound which \emph{does not have an additive $d\log T$ term}, but is instead of the form $O(d\log^2 \rho +\log T)$, with $\rho$, one of the above-mentioned parameters, potentially much smaller than $T$.

\info{first result}
{\bf Mechanism~\ref{alg:histquery}:} Our first mechanism  gives new parameterized upper bounds on the additive error in the insertions-only model for computing a large class of histogram-based queries that include $\histogram$, \maxsum\ and \minsum\footnote{\minsum: Return the minimum column sum.}, \sumselect, as well as $\quantile{p}$\footnote{$\quantile{p}$: Return the $p$-th quantile column sum.} and \topk\footnote{\topk: Return the $k$-th largest column sums.}, where the parameter depends on the queries answered.
More specifically, the additive error of the mechanism is $O\left({d\log^2 (d\cmax)+\log T}\right)$, where $\cmax$ is the \emph{maximum query value for the given input data at any time step}. Thus
if $\cmax =o(2^{\sqrt{\log T}})$,
our result breaks the $\Omega(d \log T)$ bound and is better than what can be achieved by the previous mechanisms, even when the maximum column sum $n_{\max} = \Omega(T)$. Our mechanism does not need to be given $\cmax$ or $T$ at initialization. The lower bound of $\Omega(d +\log T)$ mentioned earlier implies that the dependency on $d$ or $\log T$ cannot be removed.

\info{motivate \\ monotone \\ hist queries}
More generally, we define a class of real-valued queries, called \emph{monotone histogram queries},  that subsumes the queries mentioned above.
These queries are \emph{monotonically increasing} when a new (nonnegative) row is added, and have low \emph{sensitivity}.
Informally, a query's \emph{sensitivity} is the maximum difference between its output values on two neighboring streams (see Def.~\ref{def:sensitivity}).

\begin{definition}[Monotone histogram query]
\label{def:monotone_hist_query}
A function $q:\{\{0,1\}^d\}^*\rightarrow \mathbb{R}_{\ge 0}$  is a \emph{monotone histogram query} if it is monotonically increasing; has sensitivity $\le 1$; and depends only on the input column sums\footnote{symmetrically, this also works for monotonically decreasing functions in the deletions-only setting.}.
\end{definition}

\info{examples \\ of MHQ}
Each individual coordinate sum satisfies this definition, as do all the
queries described above. \histogram\ can be obtained with $d$ monotone histogram queries, namely, each coordinate sum.

\begin{theorem}
\label{thm:intro}
Let $x=x^1, \ldots, x^T$ be an insertions-only stream, and $q_1, \dots, q_m$ be $m$ monotone histogram queries. Mechanism~\ref{alg:histquery} is $\eps$-differentially private, and answers $(q_1(x^1,\dots,x^t), \dots, q_m(x^1,\dots,x^t))$ at all time steps $t$ with a bound on the $\ell_{\infty}$-error of $$O \left(\left(d \log^2 (dm\cmax/\beta) + m\log (m\cmax/\beta) +\log T\right)\epsilon^{-1}\right)$$ that holds with probability $\ge 1-\beta$ simultaneously over all time steps, where $\cmax=\max_{k \in [m]} q_k(x)$. Neither $T$ nor $\cmax$ need to be given as input to the mechanism.
\end{theorem}

\info{improvements \\ over prior \\ work}
Answering a single monotone histogram query  using the best existing $\epsilon$-dp mechanisms requires computing a full histogram, which gives error
$O(d\log^2(dn_{\max})+d\log T)$.
Theorem \ref{thm:intro} improves on these bounds in three significant ways: (1) The $d\log T$ term is replaced by a $\log T$ term, giving the first mechanism that achieves a $o(d \log T)$ bound on sparse outputs. (2) The polylogarithmic dependency of $d\log^2(dn_{\max})$ on the maximum column sum $n_{\max}$ is reduced to a polylogarithmic dependency of $d\log^2(d\cmax)$ on the \emph{maximum query value} $\cmax$.
For monotone histogram queries, we have $\cmax \le n_{\max}$, and $\cmax$ could be much smaller than $n_{\max}$, when computing, say, the minimum column sum or the median column sum. This is true, for example, on streams with a power-law distribution of the column sums.
(3) We can answer up to $m$ queries \emph{without incurring an extra additive $O(m\log T)$ error}, which would happen when using standard composition.
We extend these results to natural-numbered inputs, to $(\eps, \delta)$-dp, and to different neighboring definitions in Section~\ref{sec:extensions}.

\info{2nd result}
{\bf Mechanism~\ref{alg:few_switches_hist}:}
Our second mechanism gives new parameterized upper bounds in the turnstile model for continual counting and \histogram. This mechanism has an additive error of $O(d \log^2 (dK) + \log T)$, where $K$ is the \emph{number of fluctuations}, which is the number of time steps where $x^t \neq x^{t+1}$.
 It does not need to be given $K$ or $T$ at initialization. This is the first improvement in the turnstile model since the 2011 bound of $O(\log^2 T)$ for continual counting (and $O(d \log^2 (dT) + d \log T)$ for \histogram) by \citet{DBLP:journals/tissec/ChanSS11}. Recall that the parameterized result of \citet{conf/asiacrypt/DworkNRR15} only works for insertions-only streams.

\info{motivate \\ numswitch}
Our result generalizes the bound of \citet{conf/asiacrypt/DworkNRR15} in two regards: (1) to possibly negative-valued inputs, and (2) to $d$-dimensional inputs.
Our bound improves or matches the additive error of \citet{conf/asiacrypt/DworkNRR15} on insertions-only streams, since $K\leq 2dn_{\max}$. However, $K$ might be considerably smaller: a data stream that contains $n/2$ ones followed by $n/2$ zeros has $K = 1$, while $n_{\max} = T/2$. Moreover, real world data often show strong time correlations, leading to a small value of $K$. Examples include recommendation systems, where movies or products that are popular at a given time are likely to be rated consecutively by more people, and outlier monitoring processes, where many of the generated reports are identical (when nothing special has happened).
\begin{theorem}
\label{thm:introNew}
Let $x=x^1,\dots, x^T$ be a turnstile stream.  Let $K$ be the total number of times $x^{t}\neq x^{t+1}$, for all $t <T$.
Mechanism~\ref{alg:few_switches_hist} is $\eps$-differentially private, and outputs an estimate of the histogram at all time steps $t$ with $\ell_{\infty}$-error bound $O\left( \left(d \log^2 (dK/\beta) + \log T\right) \epsilon^{-1} \right)$ that holds with probability at least $1-\beta$ simultaneously over all time steps. Neither $T$ nor $K$ need to be given as input to the mechanism.
\end{theorem}

Our result also gives new insights about stronger continual counting lower bounds:
(a)  Input streams consisting of $O(2^{\sqrt{\log T}})$ non-zero entries will not lead to a stronger lower bound for continual counting~\citep{conf/asiacrypt/DworkNRR15}.
(b) Even further, the input streams would need to change frequently, i.e.,~$\omega(2^{\sqrt{\log T}})$ times, unlike the streams used in the current lower bounds of ~\citep{Dwork2010,DBLP:conf/colt/Cohen0NSS24}.

Note that tight lower bound in $K$, i.e., a lower bound of $\Omega(\log^2 K+\log T)$ for $d=1$ that holds for all values of $K$ would also imply a lower bound of $\Omega(\log^2 T)$ for continual counting (since $K$ could be as large as $T$), which is a major open problem in the area.

\section{Preliminaries}
\label{sec:prelim}

\paragraph{\bf Notation} We denote the set $\{1,\dots,n\}$ by $[n]$.

\paragraph{\bf Continual Release Model} In the continual release model, at every time step $t$, we add an element $x^t \in \universe^d$ to the current data set. Note that for simplicity, we define the model as adding $x^t$ to the data set; in the turnstile model we allow $x^t$ to have negative entries, capturing both insertions and deletions.
The entire stream of insertions is of length $T$, and $T$ is not given as input to the mechanism.

\begin{definition}\label{def:neighbouring}
Two streams $x = x^1,\dots,x^T$ and $y = y^1,\dots,y^T$ with $x^t, y^t\in \universe^d$ for all $t \in [T]$ are (event-level) \emph{neighboring streams} if there exists a time step ${t^{*}} \in [T]$ such that $x^t=y^t$ for all $t\neq{t^{*}}$ and $\|x^{t^{*}}-y^{t^{*}}\|_{\infty}\leq 1$ for insertion-only streams,
and $\|x^{t^{*}}-y^{t^{*}}\|_{\infty}\leq 2$ for turnstile streams.
\end{definition}

\paragraph{Continual Release Mechanism}
A mechanism $A$ in this model receives input $x^t \in \universe$ at every time step $t$, and produces an output $a^t=A(x^1,\dots,x^t)$ which may only rely on $x^1$ to $x^t$. $A^T(x)=(a^1,a^2,\dots,a^T)$ is the collection of the outputs at all time steps $\le T$.

\begin{definition}[Differential privacy \citep{Dwork2006}]
\label{def:dp}
A randomized mechanism $A$ on a domain $\universe^T$ is \emph{$(\epsilon,\delta)$-differentially private ($(\epsilon,\delta)$-dp)} if for all $S\in \mathrm{range}(A^T)$ and all neighboring $x,y\in \universe^T$ we have
\begin{align*}
    \Pr[A^T(x)\in S]\leq e^{\epsilon}\Pr[A^T(y)\in S]+\delta.
\end{align*}
If $\delta=0$ then $A$ is \emph{$\epsilon$-differentially private ($\epsilon$-dp)}.
\end{definition}

\begin{definition}[$L_p$-sensitivity]
\label{def:sensitivity}
The \emph{$L_p$-sensitivity} of $f :\chi^*\rightarrow \mathbb{R}^k$ is $ \max_{x,y\textnormal{ neighboring}}||f(x)-f(y)||_{p} $.
If $k=1$, then this, the \emph{sensitivity of $f$}, is equal for all $p$.
\end{definition}

We use the Laplace distribution to ensure privacy.

\begin{definition}[Laplace Distribution]
The \emph{Laplace distribution} centered at $0$ with scale $b$ is the distribution with probability density function %
\begin{align*}
\densLap{b}(x)=\frac{1}{2b}\exp\left(\frac{-|x|}{b}\right).
\end{align*}
We use $X\sim \Lap(b)$ or just $\Lap(b)$ to denote a random variable $X$ distributed according to $\densLap{b}(x)$.
\end{definition}

\begin{fact}[\citealp{Dwork2006}] \label{lem:Laplacemech} Let $f:\universe^*\rightarrow \mathbb{R}^k$ be any function with $L_1$-sensitivity $\Delta_1$. Let $Y_i\sim \Lap(\Delta_1/\epsilon)$ for $i\in[k]$. The mechanism $A(x)=f(x)+(Y_1,\dots,Y_k)$
satisfies $\epsilon$-dp.
\end{fact}
\begin{fact}\label{fact:laplace_tailbound}
If $Y \sim \Lap(b)$, then $P(|Y|\geq t\cdot b)=\exp(-t)$.
\end{fact}
\begin{fact}[Simple Composition~\citep{eurocrypt/DworkKMMN06}]
\label{fact:composition_theorem}
Let $A_1 :\chi^*\rightarrow \mathrm{range}(A_1)$ and $A_2: \universe^*\times\mathrm{range(A_1)}\rightarrow \mathrm{range}(A_2)$ be $\epsilon_1$ and $\epsilon_2$-dp mechanisms resp.. Then $A_1\circ A_2$ is $\epsilon_1+\epsilon_2$-dp.
\end{fact}

\paragraph{Differential Privacy Against an Adaptive Adversary}

As a subroutine, we use a continual histogram mechanism that works in the stronger \emph{adaptive continual release model} defined by \citet{pricedpco}.
In this model, the mechanism $M$ interacts with a \emph{randomized adversarial process $Adv$} that has no restrictions on its time or space complexity. It knows the mechanism $M$ and all its inputs and outputs up to the current time step, but \emph{not} its random coin flips.
Based on this, $Adv$ has to choose the input to $M$ for the next time step.

\emph{Event-level neighboring inputs} are modelled as follows.
All time steps except one are \emph{regular}, and the adversary is allowed to adaptively determine when the special \emph{challenge} time step occurs.
At a regular time step, $Adv$ outputs one value $x^t$. At a challenge time step, $Adv$ outputs two values $x^t_{(L)}$ and $x^t_{(R)}$ such that $x^1,\dots, x^t_{(L)},x^{t+1},\dots$ and $x^1,\dots, x^t_{(R)},x^{t+1},\dots$ are neighboring (here, $||x^t_{(L)}-x^t_{(R)}||_{\infty}\leq 1$). At the challenge time step, an external oracle selects one of these two inputs and sends it to $M$. The oracle decides before the beginning of the interaction whether it sends the first or the second input to $M$.
Importantly, this decision is not known either to $Adv$ or $M$. The goal of the adversary is to determine which decision was made by the oracle, while the goal of the mechanism is to return the computed output, e.g., a histogram, such that $Adv$ does not find out which decision was made by the oracle.

\setcounter{algorithm}{-1}

\begin{algorithm}[th]
\makeatletter
\renewcommand*{\ALG@name}{Game}
\makeatother
\begin{algorithmic}[1]

\State {\bf Input:} Stream length $T\in \mathbb{N}$, $\side\in \{L,R\}$ (not known to $Adv$ and $M$)\;
\For{$t \in [T]$}
    \State $Adv$ \textbf{outputs} $\type^t\in\{\texttt{challenge}, \texttt{regular}\}$, where \texttt{challenge} is only chosen for exactly one value of $t$
    \If{$\type^t=\texttt{regular}$}
        \State $Adv$ \textbf{outputs} $x^t\in \universe$ which is sent to $M$\EndIf
    \If{$\type^t=\texttt{challenge}$}
        \State $Adv$ \textbf{outputs} $(x^t_{(L)},x^t_{(R)})\in \universe^2$
        \State $x^t_{(\side)}$ is sent to $M$\EndIf
    \State $M$ outputs $a_t$
\EndFor
\end{algorithmic}
\caption{Privacy game $\Pi_{M,Adv}$ for the adaptive continual release model}
\label{appalg:adaptivemodel}
\end{algorithm}

More formally the relationship between $Adv$ and $M$ is modeled as
a game between adversary $Adv$ and mechanism $M$, given in Game~\ref{appalg:adaptivemodel}.

\begin{definition}[Differential privacy in the adaptive continual release model \citep{pricedpco}]
\label{appdef:adaptive}
Given a mechanism $M$ the \emph{view} of the adversary $Adv$ in game $\Pi_{M,Adv}$ (Game~\ref{appalg:adaptivemodel}) consists of $Adv$'s internal randomness, as well as the outputs of both $Adv$ and $M$. Let $V_{M,Adv}^{(\side)}$ denote $Adv$'s view at the end of the game run with input $\side\in\{L,R\}$. Let $\mathcal{V}$ be the set of all possible views. Mechanism $M$ is \emph{$(\epsilon,\delta)$-differentially private in the adaptive continual release model} if, for all adversaries $Adv$ and any $S\subseteq\mathcal{V}$,
\begin{align*}
\Pr(V_{M,Adv}^{(L)}\in S)\leq e^{\epsilon}\Pr(V_{M,Adv}^{(R)}\in S)+\delta
\end{align*} and
\begin{align*}
\Pr(V_{M,Adv}^{(R)}\in S)\leq e^{\epsilon}\Pr(V_{M,Adv}^{(L)}\in S)+\delta.
\end{align*}
We also call such a mechanism \emph{adaptively $(\epsilon,\delta)$-differentially private}.
\end{definition}

\citet{neurips2022} shows that $\epsilon$-dp under continual release implies $\epsilon$-dp against an adaptive adversary, which is not true for $(\epsilon,\delta)$-dp in general.
\begin{fact}[Prop 2.1 of \citet{neurips2022}]\label{fact:epsadaptive}
Every mechanism that is $\epsilon$-differentially private in the continual release model is also $\epsilon$-dp in the adaptive continual release model.
\end{fact}

\paragraph{\bf Continual Counting and Histogram}

The additive error bound of the $\epsilon$-dp continual counting mechanism by \citet{DBLP:journals/tissec/ChanSS11}
is $O(\log^2(T/\beta) \epsilon^{-1})$ with probability $1 - \beta$ for \emph{all time steps $t \in [T]$ simultaneously}, where $T$ is not given as input to the mechanism. The inputs are allowed to be integers, and neighboring is as in Definition~\ref{def:neighbouring}.
A simple mechanism for continual histogram is using $d$ binary counting mechanisms%
, one per column. With Facts \ref{fact:composition_theorem} and \ref{fact:epsadaptive}, this yields:%
\begin{fact}
\label{fact:counting}
There is an $\epsilon$-dp mechanism for continual histogram in the adaptive model whose $\ell_\infty$-error is bounded by $O(d\log^2(dT/\beta)\epsilon^{-1})$ with probability $1-\beta$.
\end{fact}
This gives a blackbox reduction from $\epsilon$-dp continual histogram to $\epsilon$-dp continual counting. We show in Appendix~\ref{appsec:lowerbound} that any continual histogram mechanism constructed this way must have an error of $\Omega(d\log T)$:
\begin{lemma}
Let $A$ be any $(\epsilon/d)$-dp continual counting mechanism. Then the histogram mechanism $H$, defined by running $A$ independently for each coordinate, must have an error of $\Omega( d \log (T) \epsilon^{-1})$ at some time step $t\leq T$ with constant probability.
\end{lemma}

\paragraph*{Probability Preliminaries}

\begin{lemma}\label{applem:sum_of_lap}
Let $Y_1,\dots,Y_k$ be independent variables with distribution $\Lap(b)$ and let $Y=\sum_{i=1}^k Y_i$. Then
\begin{align*}
    P(|Y|>2b\sqrt{2\ln(2/\beta_S)}\max(\sqrt{k},\sqrt{\ln(2/\beta_S)})\leq \beta_S.
\end{align*}
\end{lemma}
\begin{proof}
Apply Corollary 12.3 of \citet{journals/fttcs/DworkR14} to $b_1=\dots=b_k=b$.
\end{proof}

\begin{restatable}{lemma}{UnionBound}
\label{applem:UnionBound}
For a random variable $X \sim D$, if $\Pr[|X|>\alpha]\le \beta$, then for $X_1, X_2, \ldots, X_k \sim D$ i.i.d., we have $\Pr[\max_i |X_i| > \alpha] \le k \cdot \beta$.
\end{restatable}

We use $f_X(x)$ to denote the probability density function of a continuous random variable $X$.
For our privacy proofs, we repeatedly use the fact that if $X$ and $Y$ are independent random variables with joint probability density function $f_{X,Y}(x, y)$, then $f_{X,Y}(x, y) = f_X(x) \cdot f_Y(y)$. Thus for any event $A(X, Y)$, we have
\[
\int_{x, y} \mathds{1}[A(x, y)] f_{X, Y}(x, y) dxdy
= \int_y \Pr_X[A(X, y)] f_Y(y) dy
\]

{
\newcommand{\Klt}{\ensuremath{\mathrm{Thresh}_\ell^t}}
\renewcommand{\Thresh}{D}
\renewcommand{\Ki}{\ensuremath{\mathrm{Thresh}_k}}
\section{Histogram Queries Parameterized in Maximum Query Output}\label{sec:hq}

\info{What is \\ mech1 for}
Mechanism~\ref{alg:histquery} is designed to answer $m$ monotone histogram queries with output-sensitive error on insertion-only streams.
\info{high-level \\ goal}
It consists of two main parts, a \emph{partitioning} mechanism %
and a black-box \emph{histogram} mechanism $H$. %
The goal of the partitioning mechanism is to sparsify
the input stream provided to the histogram mechanism $H$ by partitioning the input stream into \emph{intervals}. It batches consecutive inputs to Mechanism~\ref{alg:histquery} together into an interval, and combines these inputs into a single input to $H$.

In more detail, the algorithm keeps parameters $c_i$ and $s_i$ for each coordinate, which keep the current estimate of the column sum within the current interval, and the total column sum, respectively. On an input $x^t$, it updates these values (line~\ref{line:setcands}). It then tests if any of the queries on $s$ crosses a threshold (line~\ref{line:ifcross}). If not, it returns the output from the previous round. If it does, we insert the $c_i$ values into the blackbox histogram algorithm $H$ (line~\ref{line:insertH}). It then updates the thresholds, the parameters $s_i$, and the output (lines~\ref{line:updateThresh}-\ref{line:update_out}). Note that we keep a separate threshold for each query, and they are updated differently depending on whether or not the query answer was close to the threshold in this round (lines~\ref{line:threshold}-\ref{line:threshupd}).

\info{choice of \\ parameters}
The parameters and thresholds are chosen to minimize the additive error:
The longer the intervals, the smaller the error from the histogram mechanism (since it has fewer insertions), and the larger the error \emph{within} an interval (since the same output is used for all time steps within an interval). The parameters of the mechanism ($\Kit, D_j^t$, and $C_j^t$ for example) are chosen with the goal of balancing these two kinds of error.

\begin{algorithm}[!htb]
\begin{algorithmic}[1]
\State{\bf Input:}{ Stream $x^1, x^2, \ldots \in \{0,1\}^d$, an adaptively $\epsilon$-differentially private continual histogram mechanism $H$, failure probability $\beta$, additive error bound $\errgen{t, \beta}$ that holds with probability $\ge 1 - \beta$ for the output of $H$ at time step $t$.}
\State{\bf Output:}{ Estimate of $q_k(\histshort(t))$ for all $k\in [m]$ and all $t\in\mathbb N$}
\LComment{Initialization}
\State Initialize $H$
\State $\beta' = 6\beta/\pi^2$, $\beta_t = \beta'/t^2$ for any $t \in \mathbb N$\;
\State $\Kione \leftarrow 3\epsilon^{-1}(12\ln(2/\beta')+6\ln(6/\beta')+m\ln(6m/\beta'))+3\cdot\errgen{1,\beta'/6}$ for all $k\in [m]$
\State $c_i \leftarrow 0$ for all $i \in [d]$ \Comment{column sum within interval}
\State $s_i \leftarrow 0$ for all $i \in [d]$ \Comment{histogram estimate}
\State $j \leftarrow 1$ \Comment{number of intervals}
\State $\tau_1 \leftarrow$ \taurv
\State $\textrm{out} \leftarrow \left( q_1(\textbf{0}), q_2(\textbf{0}), \ldots, q_m(\textbf{0}) \right) $ \Comment{current output}
\LComment{Process the input stream}
\For{$t \in \mathbb{N}$}
    \State $c_i \leftarrow c_i + x_i^t$, $s_i \leftarrow s_i + x_i^t$ for all $i \in [d]$\label{line:setcands}
    \LComment{Set parameters}
    \State $\amut \leftarrow$ \amutval
    \State $\atauj \leftarrow$ \ataujval
    \State $\agammaj \leftarrow 3 \epsilon^{-1} m \ln (6m/\beta_j)$
    \State $\ahj \leftarrow$ \ahjval
    \State $C_j^t \leftarrow \amut + \atauj + \agammaj$, $\T_j^t \leftarrow 3 ( C_j^t + \ahj )$

    \LComment{Test for threshold}
    \State $\mu_t \leftarrow$ \murv \label{line:mu}
    \If{\label{line:ifcross}$\exists$ $k\in[m]:$ $q_k(s)+\mu_t>\Kit + \tau_j$}
        \LComment{Close the current interval}
        \State insert $(c_1,\dots,c_d)$ into $H$\label{line:insertH}
        \State reset $c_i \leftarrow 0$ for all $i \in [d]$
        \For{$k\in [m]$}
           \State $\gamma_k^j \leftarrow$ \gammarv\label{line:gamma}
             \LComment{if $q_k(s)$ is ``close'' to threshold, increase threshold}
            \If{${q_k}(s) + \gamma_k^j>\Kit-C_j^t$\label{line:threshold}}
               \State $\Kit \leftarrow \Kit +\T_j^t$\label{line:threshupd}
                \EndIf
        \EndFor
        \State$j \leftarrow j+1$
        \LComment{update threshold for the new interval}
        \State$\Kit \leftarrow \Kit - \T_{j-1}^t + \T_j^t\ \forall k \in [m]$ \label{line:updateThresh}
        \State$\tau_j \leftarrow $ \taurv\label{line:tau} \Comment{pick fresh noise}
      \State  $(s_1,\dots,s_d) \leftarrow$ output$(H)$\label{line:outputH}
       \State $\textrm{out} \leftarrow (q_1(s),\dots,q_m(s))$\label{line:update_out}
    \EndIf
    \State\textbf{output} out
  \State  $\Kitplusone \leftarrow \Kit - \T_{j}^t + \T_j^{t+1}$ $\forall\ k \in [m]$
\EndFor
\end{algorithmic}
\caption{Mechanism for answering $m$ histogram queries parameterized in the maximum query output.}
\label{alg:histquery}
\end{algorithm}

We want to point out two main differences in our approach compared to previous work, which are due to two difficulties: first, we compute non-linear queries on high-dimensional input data,
and second, we want to break the $\Omega(d \log T)$ barrier. We explain first the issues that arise and then how we overcome them.

We explain the first  difficulty for the query \minsum, but it applies correspondingly to other monotone histogram queries as well. When $d=1$ (continual counting as in \citet{conf/asiacrypt/DworkNRR15}), the partitioning mechanism does not depend on the output of the black-box counting/histogram mechanism. This is because the continual count over a stream is equal to the sum of the continual counts of all intervals. This does not hold for non-linear queries on higher dimensional inputs
because the input stream cannot be decomposed into intervals for queries like \minsum, i.e., the \minsum\ of the entire stream cannot be obtained from knowing just the \minsum\ value of each interval.
Instead, the partitioning algorithm requires an estimate of the current \minsum\ value at every time step.

Since we do not want the computation during the current interval to depend on the private data from prior intervals, we reuse the last output of $H$, as it is a privatized number, in order to keep a running estimate of \minsum.
This, however, leads to the following technical challenge:
The partitioning mechanism \emph{depends on the outputs of the histogram mechanism of prior intervals}, and the input to the histogram mechanism depends on the output of the partitioning mechanism, and, hence, on the prior output of the histogram mechanism.
Furthermore, given two neighboring streams to Mechanism~\ref{alg:histquery}, the input to the black-box histogram mechanism that is generated by the partitioning mechanism might not necessarily be neighboring streams (consider the case where two wildly different partitions are created on two neighboring streams, leading to inputs to the histogram mechanism that are very far apart). Thus, we cannot use a simple composition theorem to show privacy for the combined mechanism.

To overcome this difficulty, we use a continuous histogram mechanism that is differentially private \emph{even if the inputs are chosen adaptively}.
We note that \emph{concurrent} composition theorems as given in, e.g., \citet{DBLP:conf/ccs/HaneySTVVX023}, cannot be used here in a black-box manner, and explain the reasons in more detail in Appendix~E.
Adaptive differential privacy of the continuous histogram mechanism allows us to separate the privacy loss incurred by the partitioning mechanism from that of the histogram mechanism.

The second difficulty relates to the $\Omega(d \log T)$ barrier.
A na\"ive partitioning technique is to maintain independent thresholds for each query, and to use the sparse vector technique (SVT) separately for each query to check if the query answer is larger than its threshold. If a query answer is larger than the threshold, an interval is closed and the thresholds are increased accordingly.
Since this involves interacting with private data at \emph{each} time step for \emph{every} query, this approach incurs the $\Omega(d \log T)$ barrier of all prior work.

\info{partition mech}
To overcome this, we design a partitioning mechanism that works as follows: it checks first if \emph{there exists} a query that crosses a certain predefined threshold value (line~\ref{line:ifcross}). If not,
then the mechanism adds the current input to the batch and \emph{does not close} the current interval. The output that was used for the previous time step is reused.

\info{query crosses \\ threshold}
If there exists a query crossing the threshold, then the mechanism \emph{closes} the current interval, sends the batched input to $H$, and initializes a new interval. %
At this point, the mechanism has privately determined that \emph{there exists} a query that crosses the threshold. However, this information is not enough to update the thresholds, since we need to also need to privately determine the \emph{identities} of all the queries that cross the thresholds, which we do next.

\info{update \\ threshold}
Finally, the mechanism checks each query independently and privately if its threshold needs to be updated (line~\ref{line:threshold}). The parameters are chosen such that at least one query has its threshold updated at the end of each interval.

For the utility proof of \citet{conf/asiacrypt/DworkNRR15}, it was enough to show that their choice of threshold was larger than the standard deviation of their Laplace random variables. Due to the interplay between several submechanisms, our utility proof requires a more involved analysis.

\begin{theorem}
\label{thm:mech1}
Let $H$ be any $(\epsilon/3)$-dp continual histogram mechanism with $\eps > 0$, and let $q_1, \dots, q_m$ be any $m$ monotone histogram queries.
Mechanism~\ref{alg:histquery} satisfies $\epsilon$-dp. If $H$ is the mechanism from Fact~\ref{fact:counting}, then on input $x$, Mechanism~\ref{alg:histquery} has additive error
\[
\eaccvalue
\]
at all $t \in [T]$ simultaneously with probability $1-\beta$, where $\cmax = \max_{k \in [m]} \max_{t \in [T]} q_k(x^1, x^2, \dots, x^t)$.
\end{theorem}
 }

\section{Histogram Parameterized in the Number of Fluctuations}
\label{sec:histswitches}

\begin{algorithm}
\begin{algorithmic}[1]

\State {\bf Input: }{Stream $x^1,x^2,\ldots\in\{-1,0,1\}^d$, an $\epsilon/3$-differentially private continual histogram mechanism $H$, failure probability $\beta$, additive error bound $\err(t,\beta)$ that holds with probability $\geq 1-\beta$ for the output of $H$ at time step $t$.}
\State {\bf Output: }{Estimate $\histshort(t)$ at all $t\in \mathbb{N}$}
\LComment{Initialization of all parameters}
\State Initialize $H$\;
\State $\beta'=6\beta/\pi^2$, $\beta_t=\beta'/t^2$ for any $t\in\mathbb{N}$\;
\State $j\leftarrow 1$ \Comment{number of intervals}
\State $c_i\leftarrow 0$ for all $i \in [d]$ \Comment{column sum within interval}
\State $\mode_i\leftarrow 0$ for all $i\in[d]$
\State $t_{\diff}\leftarrow 0$ \Comment{length of current interval}
\State $\tau_1\leftarrow\Lap(9/\epsilon)$, $\tau_2\leftarrow\Lap(9/\epsilon)$\;

\State $H_{\out}=0^d$ \Comment{initial histogram}
\LComment{Process the input stream}
\For{$t\in\mathbb{N}$}
    \State $c_i\leftarrow c_i+x_i^t$ for all $i\in [d]$\label{line:updatec}\;
    \State $t_{\diff}\leftarrow t_{\diff}+1$\label{line:updatetdiff}\;
    \State $\alpha_t\leftarrow \frac{27}{\epsilon}\log (4/\beta_t)+\frac{3d}{\epsilon}{\log (1/\beta_j)}$\;
    \State $\Thresh_{i,1}\leftarrow\mode_i\cdot t_{\diff}-2\alpha_t$, $i\in[d]$\label{line:updatethresh1}\;
    \State $\Thresh_{i,2}\leftarrow\mode_i\cdot t_{\diff}+2\alpha_t$, $i\in[d]$\label{line:updatethresh2}\;
    \State $\mu_1^t\leftarrow\Lap(18/\epsilon)$\;
    \State $\mu_2^t\leftarrow\Lap(18/\epsilon)$\;
    \If{$\min_i(c_i-\Thresh_{i,1})<\tau_1 - \mu_1^t$}\label{cond:under_thresh_hist}
            \LComment{close the current interval}\label{line:close_int1}
            \State insert $(c_1,\dots,c_d)$ into $H$\label{line:input_hist2}
            \State $H_{\out} \gets \textrm{output}(H)$\label{line:output_hist2}\Comment{update histogram}
            \For{$i\in[d]$\label{line:update_modes_1}}
                \State $\lambda_i=\Lap(3d/\epsilon)$\;
                \LComment{update modes}
                \If{$c_i+\lambda_i<\Thresh_{i,1}+\alpha_t$\label{cond:update_mode_down}}
                    \State $\mode_i\leftarrow\max\{ \mode_i-1,-1 \}$\label{line:update_mode_down}
                \EndIf
            \EndFor
            \State reset $c_i\leftarrow 0$ for all $i\in[d]$\label{line:update_hist1}\;
            \State $j\leftarrow j+1$ \;
            \State $t_{\diff}\leftarrow 0$; $\tau_1\leftarrow\Lap(9/\epsilon)$
    \ElsIf{$\max_i(c_i-\Thresh_{i,2})>\tau_2-\mu_2^t$\label{cond:over_thresh_hist}}
            \LComment{close the current interval}\label{line:close_int2}
            \State insert $(c_1,\dots,c_d)$ into $H$
            \State $H_{\out} \gets \textrm{output}(H)$\Comment{update histogram}
            \For{$i\in[d]$}\label{line:update_modes_2}
                \State $\lambda_i=\Lap(3d/\epsilon)$
                \LComment{update modes}
                \If{$c_i+\lambda_i>\Thresh_{i,2}-\alpha_t$\label{cond:update_mode_up}}
                    \State $\mode_i\leftarrow\min\{ \mode_i+1,1 \}$\label{line:update_mode_up}
                \EndIf
            \EndFor
            \State reset $c_i\leftarrow 0$ for all $i\in[d]$\label{line:update_hist2}\;
            \State $j\leftarrow j+1$\;
            \State $t_{\diff}\leftarrow 0$;  $\tau_2\leftarrow\Lap(9/\epsilon)$\label{line:end_of_update}\EndIf
    \State \bf{output} $H_{\out}+\mode \cdot t_{\diff}$
\EndFor
\end{algorithmic}
\caption{Mechanism for \histogram\ parameterized in the number of fluctuations.}
\label{alg:few_switches_hist}
\end{algorithm}

Mechanism~\ref{alg:few_switches_hist} is designed to answer \histogram\ in the turnstile model with a better error bound when the number of times two consecutive rows differ is small, called the \emph{number of fluctuations}, $K$.
The high-level structure of the mechanism is similar to that of Mechanism~\ref{alg:histquery}, consisting of a partitioning mechanism that interacts with a black-box histogram mechanism $H$. The partitioning mechanism batches inputs together into an \emph{interval}, and sends the combined inputs in an interval as a single input to the histogram mechanism.

Earlier, the output of Mechanism~\ref{alg:histquery} in an interval was set to the previous output of $H$, and the output remained unchanged within the interval. %
In general, when balancing the privacy-accuracy trade-off by limiting the number of time steps for which the private data is accessed, the classic strategy (as by \citet{conf/asiacrypt/DworkNRR15}) is to not update the output at all between two updates to the black-box mechanism (here, the mechanism $H$).
We go beyond this paradigm with Mechanism~\ref{alg:few_switches_hist}, by modifying the estimate \emph{even within an interval}. In particular, our histogram estimate is the sum of the last output of $H$ and a function that is linear in the length of the interval. The function is chosen such that if the input stream remains \emph{stable}, i.e. the value does not change often, we do not need to end the current interval often (here, only $O(K)$ times).

Specifically, the output of Mechanism~\ref{alg:few_switches_hist} within an interval mimics the behavior of the previous batch of updates.
If the histogram was ``significantly'' increasing (or decreasing) for a coordinate during the previous interval, then the output of Mechanism~\ref{alg:few_switches_hist} for this coordinate in the current interval adds (or subtracts) an estimate to the last output of $H$ at each time step. This corresponds to guessing the \emph{first-order derivative} or gradient of each coordinate of the histogram
within the previous interval and then using this to vary the output of the mechanism for the current interval.
When the additive error of the estimate accumulated within an interval crosses a pre-specified threshold for at least one coordinate, the current interval is closed and the guess for the next interval is recomputed based on whether the prior guess overestimated (or underestimated) the true count for the current interval.

In detail, the variable $c_i$ tracks the true count for coordinate $i$ within the current interval, and $t_{\diff}$ is used to count the number of time steps in the interval.
Crucially, we introduce the variables $\mode_i$ for $i \in [d]$, which assume values in $\{-1, 0 ,1\}$ and are used to guess the slope for each coordinate. Each $\mode_i$ is initially 0. Upon an input $x^t$, we first update $c_i$ and $t_{\diff}$ (lines~\ref{line:updatec}-\ref{line:updatetdiff}) and the thresholds (lines~\ref{line:updatethresh1}-\ref{line:updatethresh2}). We use $\mode_i\cdot t_{\diff}$ as a guess for the count within the current interval for each coordinate $i$ -- that is, if $\mode_i=1$ (or $0$ or $-1$) we guess that coordinate $i$ consists in this interval only of $1$'s  (or $0$'s or $-1$'s). In each round we check if for some $i$, the guess is too large or too small compared to $c_i$ (lines~\ref{cond:under_thresh_hist} and \ref{cond:over_thresh_hist}). If not, we output the previous histogram output plus  $\mode_i\cdot t_{\diff}$ for each coordinate $i$. If for some $i$, the guess is too large, we insert the $c_i$ values into the blackbox histogram algorithm $H$ and update its output (lines~\ref{line:input_hist2}-\ref{line:output_hist2}). We then reduce the modes of all coordinates $i$ where $\mode_i\cdot t_{\diff}$ is too large compared to $c_i$ (line~\ref{line:update_mode_down}). If the guess is too small for some $i$, we perform analogous updates (lines~\ref{cond:over_thresh_hist}-\ref{line:end_of_update}).

For our parameterized error bound, we use the following novel analysis approach: we subdivide the input stream into maximal contiguous substreams during which the input does not change, called \emph{episodes}, creating at most $K+1$ episodes in the entire stream. Within each episode, we prove that at most 9 intervals are closed with high probability (whp) as follows: the increase of $c_i$ in one time step, called the \emph{slope of coordinate $i$}, is constant within an episode and belongs to $\{ -1, 0, 1 \}$. Thus, within an episode, all $d$ coordinates can be placed into three groups according to their slope. Within each \emph{interval}, the mechanism maintains one of three different modes for each coordinate. We first show that for all coordinates with an identical slope and an identical mode, their modes are updated at the same time step whp. As intervals are only closed when the mode of some coordinate is modified, it suffices to bound the number of time steps within an episode when a mode is updated.

Whenever the mode of a coordinate changes, it holds whp that the absolute difference of its mode and its slope is reduced by 1. Further, the mode will not be updated anymore in this episode when it matches the slope of the coordinate.
Thus, if the slope of a group is, say,~1, and there is a subgroup of coordinates of the group with mode $-1$  at the beginning of the episode, then there will be at most two time steps where the mode of this subgroup changes, namely first to 0 and then to 1. Counting all subgroups
and including the potential interval closure when the episode ends results in up to nine intervals closed within an episode whp.

Thus, taking first-order information into account admits improved bounds parameterized in the number of fluctuations while also handling \emph{negative} inputs, and the novel extension of SVT allows the first use of input partitioning techniques in the turnstile model.

\begin{theorem}
\label{thm:mech2}
Let $H$ be any $(\epsilon/3)$-dp continual histogram mechanism with $\eps > 0$.
Mechanism~\ref{alg:few_switches_hist} satisfies $\epsilon$-dp. If $H$ is the mechanism from Fact~\ref{fact:counting}, then on input $x$, Mechanism~\ref{alg:few_switches_hist} has additive error
$
O((d \log^2 (d K /\beta) + \log T)\eps^{-1})
$
at all $t \in [T]$ simultaneously with probability $1-\beta$, where $K=\sum_{t \in [T]} \mathds{1}(x_t \neq x_{t-1})$.
\end{theorem}

\section{Extensions}
\label{sec:extensions}

In this section, we present extensions and applications of the results shown earlier.

In the case of Mechanism~\ref{alg:histquery},
our bounds also extend to the setting when the entries are natural numbers ($\universe = \mathbb N$).
We produce our bounds for $(\epsilon,\delta)$-dp in Appendix~\ref{appsec:histqueryed}, where the linear dependencies on $d$ and $m$ are replaced by $\sqrt{d}$ and $\sqrt{m}$ respectively, achieving a similar improvement over prior work as for $\epsilon$-dp. In the standard definition, neighboring streams $x$ and $y$ may differ in the \emph{entire} input vector at one time step,~i.e.,~there is one time step $t^*$ such that $x^{t^*}$ and $y^{t^*}$ could differ in all $d$ coordinates. If $x^{t^*}$ and $y^{t^*}$ may only differ in up to $b<d$ coordinates, or, more generally, $||x^{t^*}-y^{t^*}||_1\leq b$, then the linear dependency on $d$ is replaced by the same dependency in $b$.

We summarize the results of %
Theorem~\ref{thm:intro} applied to some widely used query functions:

\begin{corollary}
Let $x=x^1, \ldots, x^T$ be an insertions-only stream as in Definition~\ref{def:cont_hist}.
Consider the queries \histogram, \maxsum, \sumselect, \topk, \median, \minsum.
Mechanism~\ref{alg:histquery} answers the query at all time steps $t$, is $\eps$-dp, and has an $\ell_{\infty}$-error of $O \left(\left(d \log^2 (d \rho /\beta) + \log T\right)\epsilon^{-1}\right)$ with probability $\ge 1-\beta$ simultaneously over all time steps, where the parameter $\rho$ is $n_{\min}$ for \minsum, $n_{\mathrm{median}}$ for \median, and $n_{\max}$ for the rest, where $n_{\min}$, $n_{\mathrm{median}}$, and $n_{\max}$ are the minimum, median, and maximum column sum.
\end{corollary}

\info{improvements \\ over \\ prior work}
As in Mechanism~\ref{alg:histquery}, the same technique replaces the $d$ term with a $\sqrt{d}$ term for $(\epsilon, \delta)$-dp for Mechanism~\ref{alg:few_switches_hist} as shown in Appendix~\ref{appsec:histqueryed}.
Similarly, if two neighboring streams may differ only in up to $b<d$ coordinates at one time step, then the linear dependency in $d$ gets replaced by $b$.
 
\section{Conclusion}

We have presented black-box reductions from continual counting to continual histogram via partitioning mechanisms, and any improvement to continual counting immediately leads to improvements for our parameterized mechanisms as well.

In addition to providing stronger upper bounds on sparse queries and streams, our results further provide the following insights into possible approaches for stronger lower bounds for these problems. In the case of continual counting, the current lower bound of $\Omega(\log T)$ is shown using sequences that only have \emph{two} switches. Thus, Algorithm~\ref{alg:few_switches_hist} (as well as the algorithm of \citet{conf/asiacrypt/DworkNRR15}) give an error of $O(\log T)$ on these sequences.
However, the latter algorithm only works for inputs in $\{0,1\}^T$, while ours works even for inputs from  $\{-1,0,1\}^T$. It follows that in order to show a stronger lower bound for continual counting (if it exists) a sequence with a large number of switches has to be used, even if the input is from $\{-1,0,1\}^T$.

In the case of continual histograms, we show that while existing algorithms must have an $\Omega(d \log T)$ error, there is hope of removing this $d \log T$ dependency. On the hard sequences which achieve the $\Omega(d \log T)$ lower bound for existing algorithms, Algorithm \ref{alg:few_switches_hist} achieves an error of $O(d\log^2 d+\log T)$ and Algorithm~\ref{alg:histquery} achieves an error of $O(d\log ^2 d + d \log \log T + \log T)$.

\section{Future Directions}
\label{sec:future}

A major open question is to close the $O(\log^2 T)$ vs $\Omega(\log T)$ gap for continual counting and $\epsilon$-differential privacy. Additionally, for continual histogram, it would be very interesting to see if there exists an $\epsilon$-dp algorithm
with $\tO(d + \log^2 T)$ error, separating the dependence of the multiple dimensions from the %
dependence on $\log T$ for \emph{all} streams - or, if a lower bound of $\Omega(d \log T)$ for \emph{all} algorithms for continual histogram exists.%

Next, many histogram settings contain streams drawn from a specific underlying probability distribution. One could imagine a histogram algorithm that \emph{learns} from its output history, and predicts future histogram values, only updating the histogram when the error of the prediction is too large. This would extend our parameterized results to the domain of \emph{learning-augmented private algorithms}, which could reduce the observed error by a large factor in practice. Often, observed data is highly structured and not adversarial, admitting much lower error bounds than for the adversarial case.

\subsection*{Acknowledgements}
\label{sec:ack}

MH: This project has received funding from the European Research Council (ERC) under the European Union's Horizon 2020 research and innovation programme (MoDynStruct, No. 101019564)  \includegraphics[width=0.9cm]{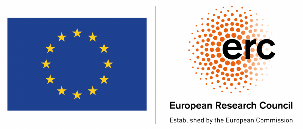} and the Austrian Science Fund (FWF) grant  \href{https://www.doi.org/10.55776/Z422}{DOI 10.55776/Z422}, grant  \href{https://www.doi.org/10.55776/I5982}{DOI 10.55776/I5982}, and grant  \href{https://www.doi.org/10.55776/P33775}{DOI 10.55776/P33775} with additional funding from the netidee SCIENCE Stiftung, 2020–2024. TAS: This work was supported by a research grant (VIL51463) from VILLUM FONDEN.

\bibliography{References}

\appendix

{
\newcommand{\Klt}{\ensuremath{\mathrm{Thresh}_\ell^t}}
\renewcommand{\Thresh}{D}
\renewcommand{\Ki}{\ensuremath{\mathrm{Thresh}_k}}
\section{Histogram Queries Parameterized in Maximum Query Output}\label{appsec:hq}

We gave an overview of how the mechanism works on an input stream in the main body. Here, we present the privacy and utility proofs of Mechanism~\ref{appalg:histquery}. We add the variable $p_j$ to the mechanism purely for the proof, to denote the end of the $j$-th interval. In particular, $[p_{j-1}, p_j)$ is the \emph{$j$-th interval}.

\begin{algorithm}[!htb]
\begin{algorithmic}[1]
\State{\bf Input:}{ Stream $x^1, x^2, \ldots \in \{0,1\}^d$, an adaptively $\epsilon$-differentially private continual histogram mechanism $H$, failure probability $\beta$, additive error bound $\errgen{t, \beta}$ that holds with probability $\ge 1 - \beta$ for the output of $H$ at time step $t$.}
\State{\bf Output:}{ Estimate of $q_k(\histshort(t))$ for all $k\in [m]$ and all $t\in\mathbb N$}
\LComment{Initialization}
\State Initialize $H$
\State $\beta' = 6\beta/\pi^2$, $\beta_t = \beta'/t^2$ for any $t \in \mathbb N$\;
\State $\Kione \leftarrow 3\epsilon^{-1}(12\ln(2/\beta')+6\ln(6/\beta')+m\ln(6m/\beta'))+3\cdot\errgen{1,\beta'/6}$ for all $k\in [m]$
\State $c_i \leftarrow 0$ for all $i \in [d]$ \Comment{column sum within interval}
\State $s_i \leftarrow 0$ for all $i \in [d]$ \Comment{histogram estimate}
\State $j \leftarrow 1$ \Comment{number of intervals}
\State $p_0 \gets 0$
\State $\tau_1 \leftarrow$ \taurv
\State $\textrm{out} \leftarrow \left( q_1(\textbf{0}), q_2(\textbf{0}), \ldots, q_m(\textbf{0}) \right) $ \Comment{current output}
\LComment{Process the input stream}
\For{$t \in \mathbb{N}$}
    \State $c_i \leftarrow c_i + x_i^t$, $s_i \leftarrow s_i + x_i^t$ for all $i \in [d]$
    \LComment{Set parameters}
    \State $\amut \leftarrow$ \amutval
    \State $\atauj \leftarrow$ \ataujval
    \State $\agammaj \leftarrow 3 \epsilon^{-1} m \ln (6m/\beta_j)$
    \State $\ahj \leftarrow$ \ahjval
    \State $C_j^t \leftarrow \amut + \atauj + \agammaj$, $\T_j^t \leftarrow 3 ( C_j^t + \ahj )$

    \LComment{Test for threshold}
    \State $\mu_t \leftarrow$ \murv \label{appline:mu}
    \If{\label{appline:ifcross}$\exists$ $k\in[m]:$ $q_k(s)+\mu_t>\Kit + \tau_j$}
        \State $p_j \gets t$\Comment{Close the current interval}
        \State insert $(c_1,\dots,c_d)$ into $H$
        \State reset $c_i \leftarrow 0$ for all $i \in [d]$
        \For{$k\in [m]$}
           \State $\gamma_k^j \leftarrow$ \gammarv\label{appline:gamma}
             \LComment{if $q_k(s)$ is ``close'' to threshold, increase threshold}
            \If{${q_k}(s) + \gamma_k^j>\Kit-C_j^t$\label{appline:threshold}}
               \State $\Kit \leftarrow \Kit +\T_j^t$\label{appline:threshupd}
                \EndIf
        \EndFor
        \State$j \leftarrow j+1$
        \State$\Kit \leftarrow \Kit - \T_{j-1}^t + \T_j^t$ for all $k \in [m]$  \Comment{update threshold for the new interval}
        \State$\tau_j \leftarrow $ \taurv\label{appline:tau} \Comment{pick fresh noise}
      \State  $(s_1,\dots,s_d) \leftarrow$ output$(H)$\label{appline:updates}
       \State $\textrm{out} \leftarrow (q_1(s),\dots,q_m(s))$
    \EndIf
    \State\textbf{output} out
  \State  $\Kitplusone \leftarrow \Kit - \T_{j}^t + \T_j^{t+1}$ $\forall\ k \in [m]$
\EndFor
\State $p_j \gets T$
\end{algorithmic}
\caption{Mechanism for answering $m$ histogram queries parameterized in the maximum query output.}
\label{appalg:histquery}
\end{algorithm}

\subsection{Privacy}
Recall that the main technical challenge to prove privacy of Mechanism \ref{appalg:histquery} is the following: The partitioning mechanism (which decides when to close an interval) \emph{depends on the outputs of the histogram mechanism for prior intervals} (unlike in \cite{conf/asiacrypt/DworkNRR15}, where the partitioning was independent of the output of the counting mechanism), and the input to the histogram mechanism depend on the output of the partitioning mechanism, and, hence, on the prior output of the histogram mechanism.
Furthermore, given two neighboring streams, the input to the histogram mechanism might not necessarily be neighboring, since the input to the histogram depends on the partitioning (consider the case where two wildly different partitions are used on two neighboring streams, leading to inputs to the histogram mechanism that are very far apart). Thus, we cannot use a simple composition theorem to show privacy for the combined mechanism.
To overcome this difficulty, we use a continuous histogram mechanism that is differentially private \emph{even if the inputs are chosen adaptively}. We then perform a careful privacy analysis to show that the interaction between the adaptively differentially private continuous histogram mechanism and the partitioning mechanism satisfies privacy. The fact that the continuous histogram mechanism is adaptively differentially private allows us to separate the privacy loss incurred by the partitioning mechanism from that of the histogram mechanism in the analysis.

\begin{lemma}
\label{applem:topkprivacy}
Let $\epsilon>0$. If $H$ is an $(\epsilon/3)$-differentially private continual histogram mechanism, then
Mechanism~\ref{appalg:histquery} satisfies $\epsilon$-differential privacy. This holds independent of the initial setting of $(s_1,\dots, s_d)$, $\Kit$, $\Thresh_j^t$, and $C_j^t$s.
\end{lemma}
\begin{proof}
Let $x$ and $y$ be two neighboring streams that differ at time $t^*$. Notice that the outputs of Mechanism~\ref{appalg:histquery} at any time step are a post-processing of
the interval partitioning and the outputs $(s_1, \dots, s_d)$ of the histogram mechanism $H$ for each interval.
Thus, to argue privacy,
we consider a mechanism
$\alg{x}$ which outputs
the interval partitions and outputs of $H$ for each interval with input stream $x$.
Let $S$ be any subset of possible outputs of $\alg$. We show that
\[
\Pr\left[ \alg{x} \in S \right]
\le e^{\epsilon} \cdot \Pr\left[ \alg{y} \in S \right]
\]
The arguments also hold when swapping the identities of $x$ and $y$ since they are symmetric, which gives us the privacy guarantee. Thus we focus on proving the inequality above.

\begin{algorithm}[!htp]
\makeatletter
\renewcommand*{\ALG@name}{Game}
\makeatother
\begin{algorithmic}[1]
\State{\textbf{Input: }}{Streams $x=x^1, x^2, \ldots, x^T \in \{0,1\}^d$ and $y=y^1, y^2, \ldots, y^T \in \{0,1\}^d$ such that $x$ and $y$ are neighboring and differ in time $t^*$, initial values $s_1,\dots,s_d$, a stream of values $\T_1,\T_2,\dots,$ a stream of values $C_1,C_2,\dots$}

\State $p_0 \leftarrow 0$, $j \leftarrow 1$

\State $c^x_i = 0$ and $c^y_i = 0$ for all $i \in [d]$

\State ChallengeOver = False

\State $\tau\leftarrow \Lap(6/\epsilon)$

\State $\widetilde{\T}_{(k)} \leftarrow \T_1 + \tau $ for all $k\in [m]$

\For{$t \in [T]$}
    \State $c^x_i = c^x_i + x_i^t$, $s_i = s_i + x_i^t$ for all $i \in [d]$
    \State $c^y_i = c^y_i + y_i^t$ for all $i \in [d]$

    \State $\mu=\Lap(12/\epsilon)$\label{appline:advnoisemutopk}
    \If{\label{appline:advKnownMaxIftopk}$\exists\ k\in[m]$: $q_k(s)+\mu>\widetilde{\T}_{(k)}$}
        \State $p_j \leftarrow t$

        \If{$p_j\geq t^{*}$ and \textnormal{ChallengeOver}=\textnormal{False}}
            \State $\type_j=\texttt{challenge}$
            \State {\bf output} $(c^x,c^y)$
            \State ChallengeOver=True
        \Else
            \State $\type_j=\texttt{regular}$
            \State {\bf output} $c^x$
        \EndIf

        \For{$k\in [m]$}
            \State $\tilde{q}_k(s)\leftarrow q_k(s)+\Lap(3m/\epsilon)$\label{appline:advnoisysitopk}
            \If{$\tilde{q}_k(s)>\T_{(k)}-C_j$\label{appline:advKnownMaxThreshtopk}}
                \State $\T_{(k)}\leftarrow\T_{(k)}+\T_j$\label{appline:update_thresh}
                \State $j \leftarrow j+1$
            \EndIf
        \EndFor

        \State $\tau \gets \Lap(6/\epsilon)$\;
        \State $\T_{(k)}\leftarrow\T_{(k)}-\T_{j-1}+\T_{j}$ for all $k \in [m]$
        \State $\widetilde{\T}_{(k)} \leftarrow \T_{(k)} + \tau$ for all $k \in [m]$
        \State reset $c^x_i \leftarrow 0$, $c^y_i \leftarrow 0$ for all $i \in [d]$
        \State receive $(s_1,\dots,s_d) \leftarrow H$\label{appline:AdvCMAssigntopk}
    \EndIf
\EndFor
\State $p_j \leftarrow T$

\caption{Privacy game $\Pi_{H,Adv(x,y)}$ for the adaptive continual release model and $m$ queries for histogram mechanism $H$}
\label{appalg:metaadversaryk}
\end{algorithmic}
\end{algorithm}

We first argue that Mechanism~\ref{appalg:histquery} acts like an adversarial process in the adaptive continual release model towards the histogram mechanism $H$. From our assumption on $H$ it then follows that the output of $H$ is $\epsilon/3$-differentially private. We will combine this fact with an analysis of the modified sparse vector technique (which determines when to close an interval) plus the properties of the Laplace mechanism (which determines when a threshold is updated) to argue that the combined mechanism consisting of the partitioning and the histogram mechanism is $\epsilon$-differentially private.

Recall that an adversary in the adaptive continual release model presented in Section~\ref{sec:prelim} is given by a privacy game, whose generic form is presented in Game~\ref{appalg:adaptivemodel}.
Due to the complicated interaction between the partitioning and $H$, the specification of such an adversarial process in our setting is given in Game~\ref{appalg:metaadversaryk}.
Call $[p_{\ell-1}, p_\ell)$ the \emph{$\ell$-th interval}.
The basic idea is as follows: Let $t^*$ be the time step at which $x$ and $y$ differ.
Conditioned on identical choices for the random variables before time step $t^*$, we have that all the intervals that the mechanism creates and also the values that the mechanism (in its role as an adversary) gives to the histogram mechanism, are identical for $x$ and $y$ before time step $t^*$.
These are regular time steps in the game.
The value for the first interval ending at or after time $t^*$ can differ and constitutes the challenge step.
All remaining intervals lead to regular steps in Game~\ref{appalg:metaadversaryk}.

Note that the end of the intervals, i.e., the partitioning of the stream, is computed by the adversary. This partitioning is based on the ``noisy'' histogram (the $s_i$ values), which are computed from the output of $H$ (which can depend on $x$ and $y$, depending on $\side$) and the values of the input stream $x$ in the current interval - for \emph{either value} of $\side$, since the adversary does not know $\side$.
We denote the adversary with input streams $x$ and $y$ by $Adv(x,y)$,
and the corresponding game, Game $\Pi_{H,Adv(x,y)}$.
Our discussion above implies that $Adv(x,y)$ does not equal $Adv(y,x)$.

The important observation from this game is that there is only one interval, i.e., only one time step for $H$, where the adversary outputs two values, and in all other time steps it outputs only one value. Also, at the challenge time step where it sends two values $c^x$ and $c^y$, these values differ by at most 1. Thus the adversarial process that models the interaction between the partitioning mechanism and $H$ fulfills the condition of the adaptive continual release model. As we assume that $H$ is $\epsilon/3$-differentially private in that model it follows that for all possible neighboring input streams $x$ and $y$ for $\Pi_{H,Adv(x,y)}$ and all possible sides $L$ and $R$ it holds that
\begin{align*}
\Pr(V_{H,Adv(x,y)}^{(L)}\in S)\leq e^{\epsilon/3}\Pr(V_{H,Adv(x,y)}^{(R)}\in S)
\end{align*}
where we use the definition of a view $V_{H,Adv(x,y)}^{(L)}$ and $V_{H,Adv(x,y)}^{(R)}$ from Definition \ref{appdef:adaptive}.
The same also holds with the positions of $x$ and $y$ switched and for $L$ and $R$ switched. Since the choice of $L/R$ merely decides whether the counts $c^x$ or $c^y$ are sent by the game to $H$, we abuse notation and specify directly which count is sent to $H$, as $V_{H,Adv(x,y)}^{(x)}$ or $V_{H,Adv(x,y)}^{(y)}$.

Recall that the view of the adversary in Game $\Pi_{H,Adv(x,y)}$ consists of its internal randomness as well as its outputs and the output of $H$ for the whole game, i.e., at the end of the game.
The behavior of $Adv(x,y)$ is completely determined by its inputs consisting of $x$, $y$, the outputs of $H$, the thresholds $\T_j^t$ and the values $C_j^t$, as well as by the functions $q_k$ and the random coin flips.
However, for the privacy analysis only
the partitioning and the output of $H$
matter since the output of Mechanism \ref{appalg:histquery} only depends on those.
Thus, we ignore the other values in the view and say that a view $V$ of the adversary $Adv(x,y)$ in Game $\Pi_{H,Adv(x,y)}$ satisfies $V\in S$,
if the partitioning and the streams of $(s_1,\dots,s_d)$ returned from $H$ for all intervals match the output sequences in $S$.
Let $C_j^t$ and $\T_j^t$ be as in the mechanism. Assume Game $\Pi_{H,Adv(x,y)}$ is run with those settings of $C_j^t$ and $\T_j^t$. By the definition of $\Pi_{H,Adv(x,y)}$, we have
\begin{align*}
\Pr(\mathcal{A}(x)\in S)&=\Pr(V_{H,Adv(x,y)}^{(x)}\in S),
\text{ and }\\
\Pr(\mathcal{A}(y)\in S)&=\Pr(V_{H,Adv(y,x)}^{(y)}\in S).
\end{align*}
We will prove below that
\begin{align}\label{appeq:viewsxyswitchedK}
\Pr(V_{H,Adv(x,y)}^{(x)}\in S)\leq e^{2\epsilon/3}\Pr(V_{H,Adv(y,x)}^{(x)}\in S).
\end{align}
Privacy then follows, since
\begin{align}\begin{split}\label{appeq:fullprivacyK}
\Pr(\mathcal{A}(x)\in S)&=\Pr(V_{H,Adv(x,y)}^{(x)}\in S)\\
\leq e^{2\epsilon/3}\Pr(V_{H,Adv(y,x)}^{(x)}\in S)
&\leq e^{\epsilon}\Pr(V_{H,Adv(y,x)}^{(y)}\in S)\\
&=e^{\epsilon}\Pr(\mathcal{A}(y)\in S).
\end{split}
\end{align}

We now prove (\ref{appeq:viewsxyswitchedK}).
Recall that when we run $Adv(x,y)$ on side $x$, the interval partitioning is created according to $x$ and the outputs of $H$. Also for each interval, the input given to $H$ is based on the counts for $x$, as we consider side $x$. When we run $Adv(y,x)$ on side $x$, then the interval partitioning is created according to $y$ and for each interval we give the counts for $x$ as input to $H$. Thus in both cases the input given to $H$ is based on the counts for $x$, and hence, to prove inequality \ref{appeq:viewsxyswitchedK}, it suffices to show that \emph{when running $Adv(x,y)$ on side $x$ and $Adv(y,x)$ on side $x$, the probabilities of getting a given partition and thresholds are $e^{2\epsilon/3}$-close.} To simplify notation, we denote running $Adv(x,y)$ on side $x$ as $\mathrm{run}(x)$, and $Adv(y,x)$ on side $x$ as $\mathrm{run}(y)$.

Recall that $[p_{\ell-1}, p_\ell)$ is the $\ell^{th}$ \emph{interval}. Denote the interval that $t^*$ belongs to as the $j$-th interval.
Note that the probabilities of computing any fixed sequence of intervals $[p_0,p_1),\dots,$ $[p_{j-2},p_{j-1})$ with $p_{j-1}<t^*$ are the same on both $\mathrm{run}(x)$ and $\mathrm{run}(y)$, since the streams are equal at all time steps before $t^*$.

We want to argue two things:
(A) fixing a particular time $\lambda>p_{j-1}$, the probability of $p_j=\lambda$ is $e^{\epsilon/3}$-close on $\mathrm{run}(x)$ and $\mathrm{run}(y)$; and
(B) the probabilities of updating the thresholds, i.e., executing line \ref{appline:update_thresh} in Game~\ref{appalg:metaadversaryk} at time $p_j$ for any subset of $[d]$, is $e^{\epsilon/3}$-close on $\mathrm{run}(x)$ and $\mathrm{run}(y)$.
Then we show that this implies that (C) all the thresholds $\Kit$ maintained by adversary are the same at the end of the interval.

The proof of (A) is similar to the privacy of the sparse vector technique (see e.g. \cite{journals/pvldb/LyuSL17}); (B) holds by a post-processing of the Laplace mechanism; and (C) follows by carefully analyzing the sequences of events and their dependencies. Before we prove these statements,
(A), (B) and (C) together imply that the probabilities that the $j$-th interval ends at the same time \emph{and} that the thresholds are updated in the same way in all intervals in run ($x$) and run($y$) are $e^{2\epsilon/3}$-close.
This implies that the probabilities $\Pr(V_{H,Adv(x,y)}^{(x)}\in S)$ and $\Pr(V_{H,Adv(y,x)}^{(x)}\in S)$ are $e^{2\epsilon/3}$-close for any subset $S$ of possible outputs. Thus, (\ref{appeq:viewsxyswitchedK}) and therefore (\ref{appeq:fullprivacyK}) follow, completing the proof.

(A) Fixing a particular time $\lambda > p_{j-1}$, we first show that the probability of interval $j$ ending at $\lambda$ (i.e., $p_j = \lambda$) is  $e^{\epsilon/3}$-close on $\mathrm{run}(x)$ and $\mathrm{run}(y)$. Fixing some notation, let $\mu_t \sim \Lap(12/\epsilon)$ and $\tau_j \sim \Lap(6/\epsilon)$ be as in the mechanism,
let $s^{t}(x)$ denote the vector of $(s_i)_{i\in[d]}$ at time $t$ for stream $x$, and $f_X$ denote the density function of the random variable $X$. For the interval $j$ to close at time $\lambda$ on $\mathrm{run}(x)$, there must exist a $k\in [m]$ with $q_k(s^{\lambda}(x)) + \mu_\lambda > \Kit + \tau_j$ at time $\lambda$, and $q_{\ell}(s^{t}(x)) + \mu_t \leq \Klt + \tau_j$ for all $p_{j-1}<t<\lambda$ and $\ell\in[m]$.

Note that conditioning on all the random variables being the same on $x$ and $y$ before $p_{j-1}$, we have that any $s_\ell$ at time $t\leq p_j$ can differ by at most 1 on $x$ and $y$. Therefore $q_\ell(s^{t}(x))$ and $q_\ell(s^{t}(y))$ can also differ by at most 1 by sensitivity of $q_\ell$. Therefore, for $p_{j-1}<t<
\lambda$, any $\ell\in[m]$ and any fixed value $z\in \mathbb{R}$ that $\tau_j$ can take, we have
\begin{align*}
    &\Pr[q_{\ell}(s^{t}(x)) + \mu_t \leq \Klt + z]\\\leq&\Pr[q_{\ell}(s^{t}(y)) + \mu_t \leq \Klt + z+1]
\end{align*}
Also, for fixed $z\in \mathbb{R}$ (resp.~$c\in \mathbb{R}$) that $\tau_j$ (resp.~$\mu_{\lambda}$) can take,
\begin{align*}
    &\Pr[q_k(s^{\lambda}(x)) + c > \Kit + z]\\\leq&\Pr[q_k(s^{\lambda}(y)) + c +2 > \Kit + z+1].
\end{align*}
Now, since $\tau_j \sim \Lap(6/\epsilon)$, we have $f_{\tau_j}(z)\leq e^{\epsilon/6}f_{\tau_j}(z+1)$. Similarly, since $\mu_\lambda\sim \Lap(12/\epsilon)$, we have $f_{\mu_{\lambda}}(c)\leq e^{2\epsilon/12}f_{\mu_{\lambda}}(c+2)=e^{\epsilon/6}f_{\mu_{\lambda}}(c)$. Now, integrating over the distributions of $\tau_j$ and $\mu_{\lambda}$ and using these properties gives $\Pr[p_j=\lambda\textnormal{ on } x]\leq e^{\epsilon/3}\Pr[p_j=\lambda \textnormal{ on } y]$.
We conclude that the probability of $p_j=\lambda$ is $e^{\epsilon/3}$-close on run($x$) and run($y$).

(B) Next, conditioned on all previous outputs of $H$ being the same and $p_j$ being equal, we argue that the probabilities of updating any subset of thresholds are close for both runs at time $p_j$.
Note that when they are updated at the same time, they are updated in the same way.
Since $q_k(s^{p_j}(x))$ and $q_k(s^{p_j}(y))$ can differ by at most $1$ for each $k\in[m]$,
adding $\gamma_k^j \sim \Lap(3m/\epsilon)$ to every
${q_k}(s^{p_j}(y))$ in line~\ref{appline:advKnownMaxThreshtopk}
ensures that the distributions of ${q_k}(s^{p_j}(x)) + \gamma_k^j$ and ${q_k}(s^{p_j}(y)) + \gamma_k^j$ are $e^{\epsilon/3}$-close for all $k\in[m]$ by composition. Since the condition in line \ref{appline:update_thresh} only depends on those, this implies that the probabilities of updating the threshold (i.e., executing line \ref{appline:update_thresh}) on any subset of $[m]$ on $\mathrm{run}(x)$ and $\mathrm{run}(y)$ are $e^{\epsilon/3}$-close.

(C)
\emph{Up to interval $j-1$:} We already argued in (A) that conditioned on all random variables being the same on $x$ and $y$ before interval $j$, the executions of run($x$) and run($y$) are identical and, thus, all thresholds are updated in the same way.
\emph{Interval $j$ and up:}
For any $\ell \ge j$ denote by $E_{\ell}$ the event that for run($x$) and run($y$), all the intervals until interval $\ell$ end at the same time step, all the thresholds $\Kit$ for $t \le p_{\ell}$ are identical, and the random variables used after time $p_{\ell}$ take the same values on both runs.
We will next argue that conditioned on event $E_{\ell}$,
event $E_{\ell + 1}$ holds.
Note that event $E_j$ holds by (B),
and by definition, run($x$) and run($y$) both use the counts from stream $x$ to compute the input for $H$.
Inductively assume that event $E_\ell$ holds.
Event $E_\ell$ implies that all intervals $\le \ell$ were closed at the same time on both runs and hence the same counts were given as input to $H$. Since (a) the streams $x$ and $y$ are identical for all $t > p_{\ell}$, (b) the thresholds and the outputs of $H$ are identical at the end of interval $\ell$, and (c) the random variables used after $p_{\ell}$ are identical (which follows from event $E_\ell$), we have that the $\ell+1$-st interval ends at the same time on both runs, and that the same thresholds are updated, and by the same amount at time $p_{\ell+1}$. This shows that event $E_{\ell+1}$ holds, as required.
\end{proof}

\subsection{Accuracy}
\label{appsec:acc}

After processing the input at time step $t$, let $\histshort^t$ be the actual histogram, $s^t$ be the value of $s$ stored by \accalg, and ${ \cmax }^t$ be the maximum query value. Suppose $t$ belongs to interval $j$, i.e., $t \in \jinterval$. Since the mechanism outputs $q_k(s^{p_{j-1}})$ at time $t$, our goal is to bound the additive error $|q_k(\histshort^t) - q_k(s^{p_{j-1}})|$ at all times $t \in \mathbb N$ and for all queries $k \in [m]$.
We do this as follows:

\begin{enumerate}
    \item Use Laplace concentration bounds to bound the maximum value attained by the random variables used by the mechanism (\bounds).
    \item Show that if query \threshcross{$k$} $\Kit$, then $q_k$ on the true histogram is not too much smaller than the threshold (Lemma~\ref{applem:accgiLB}).
    \item Show that if query \threshcross{$k$} $\Kit$, then $q_k$ on the true histogram is not too much larger than the threshold (Lemma~\ref{applem:accgiUB}).
    \item Bound the number of intervals produced by the mechanism (Lemma~\ref{applem:accnumintervals}).
    \item Use all the above to bound the error of the mechanism (Lemma~\ref{applem:acccases}).
\end{enumerate}

We define the random variables (RVs) $\mu_t, \tau_j, \gamma_k^j$ as in the mechanism.
The variables $\amut, \atauj, \agammaj$ used in the mechanism are defined such that they bound simultaneously with good probability ($\ge 1 - \beta$) the corresponding RVs.
In the rest of the section, we condition that the bounds hold on the random variables used in the mechanism.

\begin{restatable}[RV Bounds]{lemma}{accconc}
\label{applem:accconc}
There exists a histogram mechanism $H$ such that the following bounds hold simultaneously with probability $\ge 1 - \beta$ for all $t, j \in \mathbb N$ and $k \in [m]$
\begin{align*}
|\mu_t| &\le \amut, \qquad
|\tau_j| \le \atauj, \qquad
|\gamma_k^j| \le \agammaj, \qquad
\max_{t \in \jinterval} \|s^t - h^t\|_{\infty} \le \ahj \quad \forall t \in \jinterval \\
\text{where} \qquad \amut &= \amutval, \qquad
\atauj = \ataujval, \qquad
\agammaj = \agammajval, \\
\ahj &= \unbddaccj
\end{align*}
\end{restatable}

From the final bound above, we get the following lemma which bounds the error of the query values when computed on the noisy histogram $s$ stored by the mechanism.

\begin{restatable}{lemma}{acchsclose}
\label{applem:acchsclose}
\assumeconc Let $t \in [T]$ be any time step, and suppose $t \in \jinterval$. Then for all $k \in [m]$,
\[
    |q_k(s^t) - q_k(h^t)| \le \ahj.
\]
\end{restatable}

Since our output at time $t$ is $q_k(s^{p_{j-1}})$, our error is $|q_k(h^t) - q_k(s^{p_{j-1}})|$, which we bound as follows:
\begin{align*}
|q_k(h^t) - q_k(s^{p_{j-1}})|
&\le |q_k(h^t) - q_k(h^{p_{j-1}})| + |q_k(h^{p_{j-1}}) - q_k(s^{p_{j-1}})| \\
&\le |q_k(h^t) - q_k(h^{p_{j-1}})| + \ahj \tag*{(by Lemma~\ref{applem:acchsclose})} \\
&\le q_k(h^t) - q_k(h^{p_{j-1}}) + \ahj, \tag*{(since $q_k$ and $h$ are monotone and $t \ge p_{j-1}$)}
\end{align*}
our accuracy bound reduces to giving an upper bound on $q_k(h^t)$ and a lower bound on $q_k(h^{p_{j-1}})$.

We say \emph{\threshcross{$k$} at time $t$} if line \ref{appline:threshupd} of the mechanism is executed for $k$ at time $t$. Note that then $t=p_j$ for some $j$.
Our lower bound on $q_k(h^{p_j})$ will be based on the fact that \threshcross{$k$} at time $p_j$.
At time steps where \notthreshcross{$k$}, our upper bound on $q_k(h^t)$ will follow from a complementary argument to the above lower bound.

For an upper bound on $q_k(h^{p_j})$ at time steps when \threshcross{$k$}, we first show that \notthreshcross{$k$} at time $p_j - 1$ as follows: Let $\plast<\pj$ be the last time step before $p_j$ when \threshcross{$k$}, and never in between $\plast$ and $\pj$.
Then by definition of the mechanism, $\Kipj-\Kiplast=\T_j^{\pj}$. We use this to show that
$q_k$ must have increased by more than 1 between $\plast$ and $\pj$. The latter fact implies two things: first, that $j\leq \kcmax$; second, that \notthreshcross{$k$} at time $\pj-1$. The latter can be used to get an upper bound on $q_k(h^{\pj-1})$ and, by the 1-sensitivity of $q_k$, also on $q_k(h^{\pj})$.
For the first interval, there does not exist any such $\plast$ where the threshold was crossed previously. For this, we prove an auxiliary lemma that says that $p_1 > 1$, and hence no threshold was crossed at time $p_1 - 1$, and the rest of the analysis follows.

Combining the two gives an upper bound on
$q_k(h^{t}) - q_k(h^{p_{j-1}})$ of $O(\T_j^t + \amut + \atauj + \agammaj)$,
which is the crucial bound needed to upper bound $|q_k(h^{t}) - q_k(s^{p_{j-1}})|$.

Our first lemma shows that whenever \threshcross{$k$}, the query value on the true histogram is not too small compared to the threshold.

\begin{restatable}[lower bound]{lemma}{accgiLB}
\label{applem:accgiLB}
\assumeconc Let $k \in [m]$ and suppose \threshcross{$k$} at time $t=p_j$.
\[
q_k(\histshort^{p_j})\geq \Kipj- \left( \amu{p_j} + \atauj + 2\agammaj + \ahj \right).
\]
\end{restatable}

Using the strategy mentioned above, we then prove that the query value on the true histogram is never too large compared to the threshold.
Along the way, we also show that every time \threshcross{$k$}, the query value on the true histogram must increase.

\begin{restatable}[upper bound]{lemma}{accgiUB}
\label{applem:accgiUB}
\assumeconc Let $k \in [m]$ and $t \in \mathbb N$.
\[
q_k(h^{t}) < \Kit + \left( \amut + \atauj + \agammaj + \ahj + 1\right).
\]
Further, suppose \threshcross{$k$} at time $t = p_j$. Then denoting by $\plast$ the last time before $p_j$ that $k$ crossed a threshold, it also holds that
$p_j - \plast > 1$ and $|q_k(h^{p_j}) - q_k(h^{\plast})| > 1$.
\end{restatable}

We use the second part of the above lemma to bound the number of intervals created by the mechanism, where ${ \cmax }^t$ is the maximum query output at time $t$.

\begin{restatable}{lemma}{accnumintervals}
\label{applem:accnumintervals}
\assumeconc \accalg\ creates at most $\kcmaxt$ many segments upto time $t$.
\end{restatable}

\begin{restatable}{lemma}{acccases}
\label{applem:acccases}
\assumeconc Let $t \in \mathbb N$ be any time step, and suppose $t \in [p_{j-1}, p_j)$. Then \accalg\ is $\alphat_j$-accurate at time $t$, where
\[
\alphat_j = O\left( \amut + \atauj + \agammaj + \ahj \right)
\]
In particular, for all $t \in \mathbb N$, \accalg\ is $\alphat$-accurate, where
\[
\alphat = O\left( \amut + \atau{\kcmaxt} + \agamma{\kcmaxt} + \ah{\kcmaxt} \right)
.\]
\end{restatable}

The accuracy proof for
Mechanism~\ref{appalg:histquery}
then follows since we show that \bounds\ holds for the corresponding values in the mechanism, and plugging them into the above lemma.

\begin{restatable}{corollary}{eacc}
\label{appcor:eacc}
Mechanism~\ref{appalg:histquery} with
a histogram mechanism with error $\err(t, \beta)$ has error at most $$\alphat = O\left( \frac{1}{\epsilon} ( d\, \err(\kcmaxt, \beta/(\kcmaxt)^2) + m \log (\kcmaxt/\beta) + \log t \right) $$ at all time steps $t$ simultaneously with probability at least $1- \beta$. In particular, using
the histogram mechanism from Fact~\ref{fact:counting} has error at most $$\alphat = \eaccvalue$$ at all time steps $t$ simultaneously with probability at least $1 - \beta$.
\end{restatable}

\subsection{Accuracy Proofs}

\accconc*
\begin{proof}
Using \lapbound\ gives us the first three bounds below:
\begin{enumerate}
    \item $\mu_t \sim \Lap(12/\epsilon)$ satisfies $ |\mu_t| < \amutval$ with probability $\ge 1 - \beta_t/2$.
    \item $\tau_j \sim \Lap(6/\epsilon)$ satisfies $ |\tau_j| < \ataujval$ with probability $\ge 1 - \beta_j/6$.
    \item $\gamma_k^j \sim \Lap(3m/\epsilon)$ satisfies $|\gamma_k^j| \le \agammajval$ for all $k \in [m]$ simultaneously with probability $\ge 1 - \beta_j/6$.
    \item By assumption, the output of $H$ at time $p_{j-1}$ has additive error at most $\errgen{j, \beta_j/6}$ with probability at least $1 - \beta_j/6$. In particular, the histogram mechanism from Fact~\ref{fact:counting} guarantees $\errgen{j, \beta} = \unbddaccj$.
\end{enumerate}
By a union bound, all the four bounds hold at every time step with probability at least $1 - \sum_{t = 1}^{\infty} \beta_t/2 - \sum_{j = 1}^{\infty} \beta_j/2 = 1 - \beta$.
\end{proof}

\acchsclose*
\begin{proof}
This follows, since
\begin{align*}
|q_k(s^t) - q_k(h^t)|
\le \|s^t - h^t \|_{\infty}
\le \ahj
\end{align*}
where the first inequality is a consequence of $q_k$ having sensitivity one, and the second is from the \bounds.
\end{proof}

\accgiLB*
\begin{proof}
This follows from the sensitivity of $q_k$ and the fact that \threshcross{$k$} at time $\pj$.
\begin{align*}
q_k(\histshort^\pj)
&\ge q_k(s^\pj) - \ahj \tag*{(by Lemma~\ref{applem:acchsclose})} \\
&\ge {q_k}(s^\pj) + \gamma_k^j -\agammaj - \ahj \tag*{(by definition of $\agammaj$)}\\
&\ge \Ki^\pj - C_j^\pj - \agammaj - \ahj \tag*{(since \threshcross{$k$})} \\
&\ge \Ki^\pj - \amut - \atauj - 2\agammaj - \ahj \tag*{(by definition of $C_j^\pj$)}
\end{align*}
as required.
\end{proof}

\begin{restatable}{lemma}{accgiUBt}
\label{applem:accgiUBt}
\assumeconc Let $k \in [m]$ and suppose \notthreshcross{$k$} at time $t$. Then
\[
q_k(h^{t}) < \Kit + \left( \foursumj \right).
\]
\end{restatable}
\begin{proof}
Since \notthreshcross{$k$} at time $t$, either the condition in line \ref{appline:ifcross} was false or the condition in line \ref{appline:threshold} was false for $k$ at time $t$.
Thus, one of the following holds
\begin{alignat*}{3}
q_k(s^{t})
&< \Kit + \amut + \atauj
&&< \Kit + \threesumtj
&&\qquad \text{if line~\ref{appline:ifcross} was false, or} \\
q_k(s^{t})
&< \Kit - C_j^t + \agammaj
&&< \Kit + \threesumtj
&&\qquad \text{if line~\ref{appline:threshold} was false.}
\end{alignat*}
Combining this with Lemma~\ref{applem:acchsclose} gives the required bound.
\end{proof}

\begin{restatable}{lemma}{accfirstinterval}
\label{applem:accfirstinterval}
\assumeconc No interval is closed on the first time step, i.e., $p_1 > 1$.
\end{restatable}

\begin{proof}
 Note that $C^t_j = \amut + \atauj + \agammaj$.
 Thus, if the condition in line \ref{appline:ifcross} is true at time $p_j$, then the condition in line \ref{appline:threshold} is also true for some $k$.
 Said differently, whenever we end a segment, there also exists an $k$ such that \threshcross{$k$}. Using Lemma~\ref{applem:accgiLB} with $t = p_1$ gives us that
\[
 q_k(\histshort^{p_1})\geq \Ki^{p_1} - (\amu{p_1} + \atau{1} + 2\agamma{1} + \ah{1}).
\]
Note that since $\T_1^{p_1}> \amu{p_1} + \atau{1} + 2\agamma{1} + \ah{1}$,  this implies $q_k(\histshort^{p_1}) > 1$. As $q_k$ increases by at most $1$ per time step and $q_k(0,\dots,0)=0$, it follows that $p_1>1$.
\end{proof}

\accgiUB*

\begin{proof}
If \notthreshcross{$k$} at time $t$, then the bound follows from Lemma~\ref{applem:accgiUBt}.
Thus assume \threshcross{$k$} at time $t = p_j$. The first part of the claim follows if we show that \notthreshcross{$k$} at time $p_j - 1$, and $p_j-1\geq 1$, since then Lemma \ref{applem:accgiUBt} holds at time $p_j-1$ and $q_k$ has sensitivity one.
We show the claim by induction over the number of times \threshcross{$k$}.
\paragraph*{Case 1: $p_j$ is the first time \threshcross{$k$}.} Since $p_j$ is the first time \threshcross{$k$}, clearly, \notthreshcross{$k$} at time $p_j-1$. Further,
Lemma~\ref{applem:accfirstinterval} gives us that $p_j\geq p_1>1$ and therefore $p_j-1\geq 1$. %
Using Lemma~\ref{applem:accgiUBt} with $t = p_j - 1$, and the fact that $q_k$ has sensitivity one gives the required bound.

\paragraph*{Case 2: $p_j$ is not the first time \threshcross{$k$}.}
Clearly $p_j-1\ge 1$ holds in this case.
Then let $\plast$ be the last time at which \threshcross{$k$} before $p_j$.
By induction, we have for $\plast$ that
\begin{alignat*}{2}
q_k(\histshort^{\plast})
&< \Ki^{\plast} &&+ \foursum{\plast}{\last} + 1 \\
&\leq \Ki^{\pj} - \T^{\pj}_j &&+ \foursum{\pj}{j} + 1
\end{alignat*}
Since \threshcross{$k$} at time $p_j$, Lemma~\ref{applem:accgiLB} with $t = p_j$ gives
\[
q_k(\histshort^{p_{j}})\geq \Ki^{p_j} - (\amu{p_j} + \atauj + 2\agammaj + \ahj)
\]
Putting both these inequalities together, we get
\begin{align*}
|q_k(\histshort^{p_j})-q_k(\histshort^{p_{\ell}})|
&>
\left( \Ki^{p_j} - (\amu{p_j} + \atauj + 2\agammaj + \ahj\right) \\
&\ -
\left( \Ki^{\pj} - \T^{\pj}_j + (\foursum{\pj}{j} + 1) \right) \\
&= \T_j^{\pj}-\left(2\amu{p_j} + 2\atauj + 3\agammaj + 2\ahj + 1\right)
> 1,
\end{align*}
since $\T_j^t\ge3(C_j^t+\ahj)$ and $C_j^t = \threesumtj$.
As $q_k$ has sensitivity one, we have $p_j-p_{\ell}>1$, and thus, \notthreshcross{$k$} at time $p_j-1$. Lemma~\ref{applem:accgiUBt} with $t = p_{j}-1$ and the sensitivity of $q_k$ then gives the required bound.
\end{proof}

\accnumintervals*
\begin{proof}
Since $C_j^t = \threesumtj$, whenever the condition in line~\ref{appline:ifcross} is true,
then the condition in line~\ref{appline:threshold} is also true for some $k$, i.e., \threshcross{$k$}.
By Lemma~\ref{applem:accgiUB}, the query value of $q_k$ on the true histogram grows by at least one every time \threshcross{$k$}. Since ${\cmax}^t$ bounds the maximum number of times any query answer can increase before time $t$, there can be at most $\kcmaxt$ many threshold crossings for all $k \in [m]$ combined, and thus the lemma follows.
\end{proof}

\acccases*
\begin{proof}
Once we prove the first part, the second follows from Lemma~\ref{applem:accnumintervals}.
Since $t \in \jinterval$, the output of the mechanism at time $t$ is $q_k(s^{p_{j-1}})$. Thus the error at time $t$ is
\begin{align*}
|q_k(h^t) - q_k(s^{p_{j-1}})|
&\le |q_k(h^t) - q_k(h^{p_{j-1}})| + |q_k(h^{p_{j-1}}) - q_k(s^{p_{j-1}})| \\
&\le |q_k(h^t) - q_k(h^{p_{j-1}})| + \ahj \tag*{(by Lemma~\ref{applem:acchsclose})} \\
&\le q_k(h^t) - q_k(h^{p_{j-1}}) + \ahj \tag*{(since $q_k$ monotone and $t \ge p_{j_-1}$)}
\end{align*}

Our task reduces to giving an upper bound on $q_k(h^{t})$, and a lower bound on $q_k(h^{p_{j-1}})$.
We have two cases depending on whether $k$ has previously crossed a threshold. Let $\tfirst(k)$ be the first time in the whole input sequence that $k$ crosses the threshold.
\paragraph*{Case 1: $t < \tfirst(k)$.}
Since the histogram is empty before the first input arrives, $q_k(h^{p_0}) = 0$. Thus
\begin{align*}
q_k(h^{t}) - q_k(h^{p_{j-1}})
&\le q_k(h^{t}) - q_k(h^{p_0}) \\
&< \Kit + ( \foursumj + 1) \tag*{(by \ub)} \\
&= \T_{j}^t + ( \foursumj + 1) \tag*{(since $\Kit = \T_j^t$)} \\
&= O( \foursumj ) \tag*{(since $\T_j^t = 3(\foursumj)$)}
\end{align*}

\paragraph*{Case 2: $t \ge \tfirst(k)$.}
Let $\plast$ be the largest time step before $t$ when \threshcross{$k$}. Then
\begin{align*}
q_k(h^t)
&\le \Kit + (\foursumj+1) \tag*{(by \ub)} \\
\text{and} \quad q_k(h^{\plast})
&\geq \Kiplast \qquad - (\amu{\plast} + \atau{\last} + 2\agamma{\last} + \ah{\last}) \tag*{(by \lb)} \\
&\geq \Kit - \T_j^t - (\amut + \atauj + 2\agammaj + \ahj)
\end{align*}
Putting these together, we get
\begin{align*}
q_k(h^t) - q_k(h^{p_{j-1}})
&\le q_k(h^t) - q_k(h^{\plast}) \\
&= O( \foursumj ) \tag*{(since $\T_j^t = 3(\foursumj)$)}
\end{align*}
which proves the lemma.
\end{proof}

\eacc*
\begin{proof}
\bounds\ and Lemma~\ref{applem:acccases} together give us that Mechanism~\ref{appalg:histquery} is $\alphat$-accurate at time $t$, where
\begin{alignat*}{2}
\amut
&= O\left( \epsilon^{-1} \log (2t/\beta) \right), \quad \quad \,\,
\atauj
&&= O\left( \epsilon^{-1} \log (j/\beta) \right), \\
\agammaj
&= O\left( \epsilon^{-1} m \log (mj/\beta) \right), \quad
\ahj
&&= O( \errgen{j, \beta/ j^2} )
\end{alignat*}
Since $\alphat = O\left( \foursumj \right)$ with $j \le \kcmaxt$, this gives the lemma.
For the histogram mechanism from Fact~\ref{fact:counting},
\begin{align*}
\ahj
&= O\left(\epsilon^{-1} d \cdot \left( \log (j) \log(dj/\beta) + (\log j)^{1.5} \sqrt{\log (dj/\beta)}  \right)  \right) \\
&= O\left( \epsilon^{-1} d \log^2 (dj/\beta) \right), \
\end{align*}
which gives
\begin{align*}
\alphat = \eaccvalue
\end{align*}
as claimed.
\end{proof}

\subsection{Extensions}

For $(\epsilon, \delta)$-dp, we use an adaptively differentially private continual histogram mechanism $H$ and the Gaussian mechanism for $\gamma_k^j$, which gives an error bound of
\[
\alphat = O \left( \epsilon^{-1}  \log(1/\delta) \cdot \left( \sqrt{d} \log^{1.5} (dm{\cmax}^t/\beta) + \sqrt{m} \log (m{ \cmax }^t/\beta)  + \log t \right)\right)
\]
for $(\epsilon, \delta)$-differential privacy. We present the technical details in Appendix~\ref{appsec:histqueryed}.

 }

\section{Histogram Parameterized in the Number of Fluctuations}
\label{appsec:histswitches}

As earlier, we gave an overview of how the mechanism works on an input stream in the main body. Here, we present the privacy and utility proofs of Mechanism~\ref{appalg:few_switches_hist}. We add the variable $p_j$ to the mechanism purely for the proof, to denote the end of the $j$-th interval. In particular, $[p_{j-1}, p_j)$ is the \emph{$j$-th interval}.

\begin{algorithm}
\begin{algorithmic}[1]

\State {\bf Input: }{Stream $x^1,x^2,\ldots\in\{-1,0,1\}^d$, an $\epsilon/3$-differentially private continual histogram mechanism $H$, failure probability $\beta$, additive error bound $\err(t,\beta)$ that holds with probability $\geq 1-\beta$ for the output of $H$ at time step $t$.}
\State {\bf Output: }{Estimate $\histshort(t)$ at all $t\in \mathbb{N}$}
\LComment{Initialization of all parameters}
\State Initialize $H$\;
\State $\beta'=6\beta/\pi^2$, $\beta_t=\beta'/t^2$ for any $t\in\mathbb{N}$\;
\State $j\leftarrow 1$ \Comment{number of intervals}
\State $p_0 \gets 0$
\State $c_i\leftarrow 0$ for all $i \in [d]$ \Comment{column sum within interval}
\State $\mode_i\leftarrow 0$ for all $i\in[d]$
\State $t_{\diff}\leftarrow 0$ \Comment{length of current interval}
\State $\tau_1\leftarrow\Lap(9/\epsilon)$, $\tau_2\leftarrow\Lap(9/\epsilon)$\;

\State $H_{\out}=0^d$ \Comment{initial histogram}
\LComment{Process the input stream}
\For{$t\in\mathbb{N}$}
    \State $c_i\leftarrow c_i+x_i^t$ for all $i\in [d]$\;
    \State $t_{\diff}\leftarrow t_{\diff}+1$\;
    \State $\alpha_t\leftarrow \frac{27}{\epsilon}\log (4/\beta_t)+\frac{3d}{\epsilon}{\log (1/\beta_j)}$\;
    \State $\Thresh_{i,1}\leftarrow\mode_i\cdot t_{\diff}-2\alpha_t$, $i\in[d]$\;
    \State $\Thresh_{i,2}\leftarrow\mode_i\cdot t_{\diff}+2\alpha_t$, $i\in[d]$\;
    \State $\mu_1^t\leftarrow\Lap(18/\epsilon)$\;
    \State $\mu_2^t\leftarrow\Lap(18/\epsilon)$\;
    \If{$\min_i(c_i-\Thresh_{i,1})<\tau_1 - \mu_1^t$}\label{appcond:under_thresh_hist}
            \State $p_j \gets t$\Comment{close the current interval}\label{appline:close_int1}
            \State insert $(c_1,\dots,c_d)$ into $H$
            \State $H_{\out} \gets \textrm{output}(H)$\Comment{update histogram}
            \For{$i\in[d]$\label{appline:update_modes_1}}
                \State $\lambda_i=\Lap(3d/\epsilon)$\;
                \LComment{update modes}
                \If{$c_i+\lambda_i<\Thresh_{i,1}+\alpha_t$\label{appcond:update_mode_down}}
                    \State $\mode_i\leftarrow\max\{ \mode_i-1,-1 \}$\label{appline:update_mode_down}
                \EndIf
            \EndFor
            \State reset $c_i\leftarrow 0$ for all $i\in[d]$\label{appline:update_hist1}\;
            \State $j\leftarrow j+1$ \;
            \State $t_{\diff}\leftarrow 0$; $\tau_1\leftarrow\Lap(9/\epsilon)$
    \ElsIf{$\max_i(c_i-\Thresh_{i,2})>\tau_2-\mu_2^t$\label{appcond:over_thresh_hist}}
            \State $p_j \gets t$\Comment{close the current interval}\label{appline:close_int2}
            \State insert $(c_1,\dots,c_d)$ into $H$
            \State $H_{\out} \gets \textrm{output}(H)$\Comment{update histogram}
            \For{$i\in[d]$}\label{appline:update_modes_2}
                \State $\lambda_i=\Lap(3d/\epsilon)$
                \LComment{update modes}
                \If{$c_i+\lambda_i>\Thresh_{i,2}-\alpha_t$\label{appcond:update_mode_up}}
                    \State $\mode_i\leftarrow\min\{ \mode_i+1,1 \}$\label{appline:update_mode_up}
                \EndIf
            \EndFor
            \State reset $c_i\leftarrow 0$ for all $i\in[d]$\label{appline:update_hist2}\;
            \State $j\leftarrow j+1$\;
            \State $t_{\diff}\leftarrow 0$;  $\tau_2\leftarrow\Lap(9/\epsilon)$\EndIf
    \State \bf{output} $H_{\out}+\mode \cdot t_{\diff}$
\EndFor
\State $p_j\gets T$
\end{algorithmic}
\caption{Mechanism for \histogram\ parameterized in the number of fluctuations.}
\label{appalg:few_switches_hist}
\end{algorithm}

\subsection{Privacy}
\label{appsec:hs_privacy}

We show $\eps$-differential privacy for neighboring streams $x$ and $y$ which are allowed to differ at one time step by $\ell_\infty$-norm $1$, i.e., there exists a $t^*$ such that $||x^{t^{*}}-y^{t^{*}}||_{\infty}\leq 1$. Then, by group privacy, the mechanism is $2\eps$-differentially private for neighboring streams that are allowed to differ by $\ell_\infty$-norm $2$.

\begin{lemma}\label{applem:few_switches_priv}
Mechanism~\ref{appalg:few_switches_hist} is $\epsilon$-differentially private.%
\end{lemma}
To prove Lemma~\ref{applem:few_switches_priv}, we note that the outputs of Mechanism~\ref{appalg:few_switches_hist} are a post-processing of three parts: 1.) the mechanism computing $p_0,p_1,\dots$; 2.) the mechanism updating the values of $\mode$; 3.) and a continual histogram mechanism. Note that in Mechanism~\ref{appalg:few_switches_hist}, 1.) and 2.) do \emph{not} depend on the outputs of 3.), and 3.) is $\epsilon/3$-differentially private by assumption. Thus, it is enough to show that the parts of the mechanism computing $p_0,p_1,\dots$ and the values of $\mode$ together are $2\epsilon/3$-differentially private, as then Mechanism~\ref{appalg:few_switches_hist} is differentially private by the composition theorem (Fact~\ref{fact:composition_theorem}).
\begin{lemma}\label{applem:few_switches_priv_1}
    The mechanism obtained by running Mechanism~\ref{appalg:few_switches_hist} and outputting only the values of $p_0,p_1,\dots$, and the values of $\mode_i$ at every time step, is $2\epsilon/3$-differentially private.%
\end{lemma}
We first provide a short proof sketch.
The partitioning mechanism basically consists of instances of AboveThreshold (see~\cite{journals/fttcs/DworkR14}) on disjoint parts of the stream, and the updating of modes is a post-processing of a Laplace mechanism. Thus, intuitively, it should be enough to use parallel composition on both the AboveThreshold mechanisms and the Laplace mechanisms, and then use sequential composition on 1.) and 2.). However, there is a technicality that the inputs and thresholds to AboveThreshold mechanisms depend on their previous outputs, and the mode updates depend on the output of the AboveThreshold mechanism, which depends on the previous mode update. Thus, we cannot use a regular parallel composition argument and thus provide the full proof.

\begin{proof}
We focus on proving that the part of Mechanism~\ref{appalg:few_switches_hist} computing $p_0,p_1,\dots$ and the sequence of $\mode$ values is $2\epsilon/3$-differentially private. Since these do not depend on the outputs of 3.) the histogram mechanism, we can then use composition to argue that Mechanism \ref{appalg:few_switches_hist} is $\epsilon$-differentially private under continual observation.

Consider any possible partitioning $P=[p_0,p_1),\dots,[p_{\ell-1},p_{\ell})$ of $[0,T)$ and any stream of $M=M^0,\dots, M^T$, where $M^t=(\mode_1^t,\dots, \mode_d^t)$ and $\mode_i^t$ is the setting of variable $\mode_i$ at time $t$. We show that the probabilities of getting $P$ and $M$ when running Mechanism~\ref{appalg:few_switches_hist} on two neighboring streams are $2\epsilon/3$-close. %
For two neighboring streams $x$ and $y$, let $t^{*}$ be the time where $x$ and $y$ differ and let $[p_{j-1},p_{j})$ be the interval in $P$ which contains $t^{*}$. Note that $||x^{t^{*}}-y^{t^{*}}||_{\infty}\leq 1$. %
Before $p_{j-1}$, all probabilities are the same; conditioning on the run of $x$ and $y$ to be identical up to time $p_{j-1}$, we show that the probabilities of closing the next interval of $y$ at $p_j$ and updating the mode of $y$ to $\mode^{p_j}_i$ are close to the probabilities of doing so on $x$. Note that the original setting of $\mode_i$ is always 0 for all $i$ on both $x$ and $y$. Let $c^t_i(x)$ and $c^t_i(y)$ denote the values of $c_i$ at time $t$ on the run of $x$ and $y$, respectively.

We compare the probabilities of closing the $j$th interval at time $p_j$ for $x$ and $y$, and doing so because the conditions in line \ref{appcond:under_thresh_hist} resp. \ref{appcond:over_thresh_hist} were fulfilled (note that we need to differentiate these two, since the values of $\mode$ get updated differently in both cases). In order to close the $j$-th interval at time $p_j$, either the condition in line \ref{appcond:under_thresh_hist} or \ref{appcond:over_thresh_hist} has to be fulfilled at time $p_j$, and neither of them can be fulfilled at any time $t\in (p_{j-1},p_j)$. First, we analyze the probabilities that at time $p_j$, condition \ref{appcond:under_thresh_hist} was fulfilled, and neither condition \ref{appcond:over_thresh_hist} nor \ref{appcond:under_thresh_hist} was fulfilled at times $t\in(p_{j-1},p_j)$ (case A). %
 Let $z_1$ be some fixed value for $\tau_1$ in the interval $(p_{j-1},p_j]$, $z_2$ some fixed value of $\tau_2$ in the interval $(p_{j-1},p_j]$, and $m_1$ a fixed value of $\mu_1^{p_j}$. We have
 \begin{itemize}
     \item for all $p_{j-1} < t<p_j$: $\Pr[\min_i c^t_i(x)+\mu_1^t\geq\Thresh_1+z_1]\leq \Pr[\min c^t_i(y)+\mu_1^t\geq\Thresh_1+z_1-1]$,
     \item for all $p_{j-1} <t<p_j$: $\Pr[\max_i c^t_i(x)+\mu_2^t\leq\Thresh_2+z_2]\leq \Pr[\max_i c^t_i(y)+\mu_2^t\leq\Thresh_2+z_2+1]$,
     \item for $t=p_j$: $\Pr[\min_i c^t_i(x)+m_1<\Thresh_1+z_1]\leq \Pr[\min_i c^t_i(y)+m_1-2<\Thresh_1+z_1-1]$.
 \end{itemize}
 Thus, the same outcome can be achieved by shifting $\tau_1$ and $\tau_2$ by at most $1$, and $\mu_1^{p_j}$ by at most $2$. By integrating in the same way as in the proof of Lemma \ref{applem:topkprivacy}, since $\tau_1$ and $\tau_2$ are distributed according to $\Lap(9/\epsilon)$ and $\mu^t_1$ is distributed according to $\Lap(18/\epsilon)$, the distributions are $\epsilon/3$-close.

 Next, we analyze the probabilities that at time $p_j$, condition \ref{appcond:over_thresh_hist} was fulfilled, and neither \ref{appcond:over_thresh_hist} nor \ref{appcond:under_thresh_hist} was fulfilled at times $t\in(p_{j-1},p_j)$, and the condition in line \ref{appcond:under_thresh_hist} was not fulfilled at time $p_j$ (case B). Similar to before, let $z_1$ be some fixed value for $\tau_1$ in the interval $(p_{j-1},p_j]$, $z_2$ some fixed value of $\tau_2$ in the interval $(p_{j-1},p_j]$, and $m_2$ a fixed value of $\mu_2^{p_j}$. We have
 \begin{itemize}
     \item for all $p_{j-1} <t\leq p_j$: $\Pr[\min_i c^t_i(x)+\mu_1^t\geq\Thresh_1+z_1]\leq \Pr[\min c^t_i(y)+\mu_1^t\geq\Thresh_1+z_1-1]$,
     \item for all $p_{j-1} <t<p_j$: $\Pr[\max_i c^t_i(x)+\mu_2^t\leq\Thresh_2+z_2]\leq \Pr[\max_i c^t_i(y)+\mu_2^t\leq\Thresh_2+z_2+1]$,
     \item for $t=p_j$: $\Pr[\max_i c^t_i(x)+m_2>\Thresh_2+z_2]\leq \Pr[\max_i c^t_i(y)+m_2+2>\Thresh_2+z_2+1]$.%
 \end{itemize}
Thus, the same outcome can be achieved by shifting $\tau_1$ and $\tau_2$ by at most $1$, and $\mu_2^{p_j}$ by at most $2$. By integrating in the same way as in the proof of Lemma \ref{applem:topkprivacy}, since $\tau_1$ and $\tau_2$ are distributed according to $\Lap(9/\epsilon)$ and $\mu^t_2$ is distributed according to $\Lap(18/\epsilon)$, the distributions are $\epsilon/3$-close.

Conditioning on ending the $j$-th interval at $p_j$ and case A resp. case B, we need to argue about the updating of the modes (lines \ref{appline:update_mode_down} and \ref{appline:update_mode_up}, respectively). Since we add $\Lap(3d/\epsilon)$ to $c^{p_j}_i(x)$ resp. $c^{p_j}_i(y)$ for all $i$, and $||c^{p_j}(x)-c^{p_j}(y)||_1\leq d$, the probabilities of any output set are $\epsilon/3$-close by the properties of the Laplace mechanism. By post-processing, the probabilities of updating $\mode_i$ to $M^{p_j}$ are $\epsilon/3$-close on $x$ and $y$. Together, the probabilities of getting $[p_0,p_1),\dots,[p_{j-1},p_j)$ and $M^0,\dots,M^{p_j}$ are $2\epsilon/3$ close on $x$ and $y$. Since the rest of the mechanism depends only on $p_j$, $M^{p_j}$, and the input streams for times $t>p_j>t^{*}$, conditioning on $p_j$ and $M^{p_j}$ the probabilities are equal. We get that the probabilities of getting $P$ and $M$ are $2\epsilon/3$-close on $x$ and $y$. This shows that the partitioning mechanism together with the mode updates is $2\epsilon/3$-differentially private.
\end{proof}

\subsection{Accuracy}
\label{appsec:hs_accuracy}

Let $K_t$ be the number of times up to time $t$ that two consecutive input data rows differ, even if they differ just in one coordinate, i.e., the number of time such that $x^{t'} \ne x^{t'+1}$ for $t' \le t$.
\begin{lemma}
\label{applem:switcherror}
With probability at least $1-3\beta$, Mechanism \ref{appalg:few_switches_hist} has error at most $O(\err(9K_t+9,\beta) + \frac{d}{\epsilon}\log (9K_t+9) + \frac{1}{\epsilon}\log (t/\beta)))$ at all time steps $t$, where $\err(\ell,\beta)$ is the error of the histogram mechanism $H$ that holds with probability at least $1-\beta$ for all length-$\ell$ prefixes of the input stream.
\end{lemma}
Similarly to the accuracy proof in Section~\ref{appsec:hq}, our proof of Lemma~\ref{applem:switcherror} builds on first bounding the values of all random variables and the error of $H$ such that all bounds hold simultaneously with probability $1-3\beta$. We call this event $E$ and condition on it. Formally,

\begin{enumerate}
    \item Let $t \in [1,T]$. With probability at least $1-\beta_t/4$, any $Y$ drawn from $\Lap(b/\epsilon)$ for any $b$ has an absolute value of at most $\frac{b\log (4/\beta_t)}{\epsilon}$.
    Thus, by a union bound,
    with probability at least $1-\beta_t$, we have $|\mu_i^t|\leq \frac{18\log (4/\beta_t)}{\epsilon}$ and $|\tau_i|\leq \frac{9\log (4/\beta_t)}{\epsilon}$ for $i=1$ and $i = 2$ at any fixed time $t$.
    Using a union bound over all time steps $t$  and observing that $\sum_{t \in [1,T]} \beta_t = \sum_{t \in [1,T]} 6 \beta'/(\pi^2 t^2) \le \beta$, it follows that  with probability at least $1- \beta$ for $i =1$ and $i=2$ and for all time steps $t$ that
    $|\mu_i^t|\leq \frac{18\log (4/\beta_t)}{\epsilon}$.
    \item With probability at least $1-\beta_j$%
    , any $Y$ drawn from $\Lap(3d/\epsilon)$ has an absolute value of at most $\frac{3d}{\epsilon}\log (1/\beta_j)$. Thus, by a union bound as above, with probability at least $1-\beta$, we have $|\lambda_i|\leq \frac{3d}{\epsilon}{\log (1/\beta_j)}$ for all values of $\lambda_i$.
    \item By the properties of $H$, with probability at least $1-\beta$, the error of $H$ after $j$ inputs is at most $\err(j,\beta)$ for all $j$.
\end{enumerate}

{\em Event $E$} is the event that all three conditions hold, which happens with probability at least $1-3\beta$.
The proof now consists of two main lemmata, Lemma~\ref{appswitchacc:item} and Lemma~\ref{applem:switch}. To show them we need:

\begin{claim}
\label{appclaim:jointup}
Conditioned on event $E$ the following hold:

1)  If at some time $t$ the condition in line \ref{appcond:under_thresh_hist} is true for some $i \in [1,d]$, then the condition in line \ref{appcond:update_mode_down} is true for all $\ell \in [1,d]$ with $c_{\ell} = c_i$ and $\mode_{\ell}=\mode_i$, i.e., $\mode_{\ell}$ is updated for all such ${\ell}$.%

2)  If at some time $t$ the condition in line \ref{appcond:over_thresh_hist} is true for some $i \in [1,d]$, then the condition in line \ref{appcond:update_mode_up} is true for all $\ell \in [1,d]$ with $c_{\ell} = c_i$ and $\mode_{\ell}=\mode_i$, i.e., $\mode_{\ell}$ is updated for all such ${\ell}$.%

\end{claim}

\begin{proof}[Proof of Claim \ref{appclaim:jointup}]
    If the condition in line \ref{appcond:under_thresh_hist} or line \ref{appcond:over_thresh_hist} is true, $t=p_j$ for some $j$. In the following, we use variable names to denote their value at time $t$ when line~\ref{appcond:under_thresh_hist} is reached.
    Further, if the condition in line \ref{appcond:under_thresh_hist} is true, then by the assumed bounds on the random variables in event $E$, there exists a $c_i$ such that
    $c_i<\Thresh_{i,1}+\frac{27\log (4/\beta_t)}{\epsilon}$. Since $|\lambda_i|\leq \frac{3d}{\epsilon}{\log (1/\beta_j)}$, we have $c_i+\lambda_i \le c_i+\frac{3d}{\epsilon}{\log (1/\beta_j)}< \Thresh_{i,1}+\frac{27\log (4/\beta_t)}{\epsilon}+\frac{3d}{\epsilon}{\log (1/\beta_j)}= \Thresh_{i,1}+\alpha_t$, by definition of $\alpha_t$. Thus, the condition in line \ref{appcond:update_mode_down} is true for $i$.  Note that this is also case for all ${\ell} \in [1,d] \setminus i$ with $c_{\ell} = c_i$ and $\mode_{\ell}=\mode_i$, since then $\Thresh_{i,1}=\Thresh_{\ell,1}$.
    Now, to show that $\mode_{i}$ actually changed, we need to show that $\mode_{i}\neq -1$ at the beginning of the round. Assume $\mode_i= -1$ at the beginning of the round. Then $c_i+\lambda_i< \Thresh_{i,1}+\alpha_t=-t_{\diff}-\alpha_t$, thus $c_i< -t_{\diff}-\alpha_t+\frac{3d}{\epsilon}{\log (1/\beta_j)}<-t_{\diff}$, which is a contradiction since $t_{\diff}\geq c_i\geq -t_{\diff}$ always. By the same argument, $\mode_{\ell}$ is updated for all $\ell\in[1,d]$ with $c_{\ell}=c_i$  and $\mode_{\ell}=\mode_i$. This proves the first part of the claim.

     Similarly, if the condition in line \ref{appcond:over_thresh_hist} is true, then by the assumed bounds on the random variables, there exists an $i$ such that
    $c_i>\Thresh_{i,2}-\frac{27\log (4/\beta_t)}{\epsilon}$.  Note that this is also case for all $\ell \in [1,d]$, $j \ne i$ with $c_j = c_i$.
    Since $|\lambda_i|\leq \frac{3d}{\epsilon}{\log (1/\beta_j)}$, we have $c_i+\lambda_i\geq c_i-\frac{3d}{\epsilon}{\log (1/\beta_j)}> \Thresh_{i,2}-\frac{27\log (4/\beta_t)}{\epsilon}-\frac{3d}{\epsilon}{\log (1/\beta_j)} = \Thresh_{i,2}-\alpha_t$, by definition of $\alpha_t$. Thus, the condition in line \ref{appcond:update_mode_up} is true for $i$.  Note that this is also case for all $\ell \in [1,d]$, $\ell \ne i$ with $c_{\ell} = c_i$ and $\mode_{\ell}=\mode_i$.
    Now, to show that $\mode_i$ actually changed, we need to show that $\mode_i\neq 1$ at the beginning of the round. Assume it was. Then $c_i+\lambda_i> \Thresh_{i,2}-\alpha_t=t_{\diff}+\alpha_t$, thus $c_i>t_{\diff}+\alpha_t-\frac{3d}{\epsilon}{\log (1/\beta_j)}>t_{\diff}$, which is a contradiction since $t_{\diff}\geq c_i\geq -t_{\diff}$ always. By the same argument, $\mode_{\ell}$ is changed, for all $\ell\in[1,d]$ with $c_{\ell}=c_i$ and $\mode_{\ell}=\mode_i$. This proves the second part of the claim.
   \end{proof}

 The next lemma shows an error bound at time $t$ on the output of Mechanism~\ref{appalg:few_switches_hist} depending on the number of intervals ($n_t$) produced by the mechanism. This follows from (A) the fact that the histogram receives $n_t$ inputs, and (B) from the fact that we can bound the additional error accumulated within an interval by~$O(\frac{d}{\epsilon}\log (n_t) + \frac{1}{\epsilon}\log (t/\beta))$.
    \begin{lemma} %
Let $n_t$ be the number of closed intervals at the end of time step $t$, i.e., if $t=p_j$ for some $j$, then $n_t=j$, and if $t\in(p_{j-1},p_j)$, then $n_t=j-1$. Conditioned on event $E$, the maximum additive error for all time steps $t'\leq t$ is $O(\err(n_t,\beta) + \frac{d}{\epsilon}\log (n_t) + \frac{1}{\epsilon}\log (t/\beta))$.\label{appswitchacc:item}\end{lemma}
\begin{proof}[Proof of Lemma \ref{appswitchacc:item}]
We differentiate two cases.
\begin{enumerate}
    \item If $t=p_j$ for some $j\leq n_t$, the output is equal to $H_{\out}$, in which case the error is at most $\err(n_t,\beta)$, by the properties of $H$.
    \item If $t\in(p_{j-1},p_j)$, $j\leq n_t+1$, let $t=p_{j-1}+t_{\diff}$ and denote by $\out^t$ the output at time $t$.
    The error at time step $t$ is given by $\max_i|\sum_{t'=1}^t x_i^{t'}-\out^t|\leq \max_i(|\sum_{t'=1}^{p_{j-1}} x_i^{t'}-\out^{p_{j-1}}|+|c_i-\mode_i \cdot t_{\diff}|$, where $c_i$, $\mode_i$ and $t_{\diff}$ correspond to the values of those variables in the mechanism at time $t$. $|\sum_{t'=1}^{p_{j-1}} x_i^{t'}-\out^{p_{j-1}}|\leq \err(n_t,\beta)$ by the properties of $H$. Thus, we bound $|c_i-\mode_i \cdot t_{\diff}|$.
    If either condition \ref{appcond:over_thresh_hist} or \ref{appcond:under_thresh_hist} would have been true, then $t=p_j$ for some $j$, a contradiction.
    Thus, both conditions were false.
    Since we conditioned on the absolute values of $\mu_1^t, \mu_2^t$ being bounded by $\frac{18\log (4/\beta_t)}{\epsilon}$ and the absolute values of $\tau_1,\tau_2$ being bounded by $\frac{9\log (4/\beta_t)}{\epsilon}$, we have $c_i-\mode_i \cdot t_{\diff}+2\alpha_t\geq -\frac{27\log (4/\beta_t)}{\epsilon}$ for all $i$ (because \ref{appcond:under_thresh_hist} is false) and $c_i-\mode_i \cdot t_{\diff}-2\alpha_t\leq\frac{27\log (4/\beta_t)}{\epsilon}$ for all $i$ (because \ref{appcond:over_thresh_hist} is false). We have $|c_i-\mode_i \cdot t_{\diff}|\leq \frac{27\log (4/\beta_t)}{\epsilon}+2\alpha_t=O(\frac{d}{\epsilon}\log (1/\beta_j) + \frac{1}{\epsilon}\log (1/\beta_t))=O(\frac{d}{\epsilon}\log (n_t/\beta) + \frac{1}{\epsilon}\log (t/\beta))$. \qedhere
\end{enumerate}
\end{proof}

   \begin{lemma} Conditioned on event $E$, at any time step $t$, no more than $\min(t, 9K_t+9)$ intervals were closed at the end of time step $t$. \label{applem:switch}\end{lemma}
   To prove Lemma~\ref{applem:switch} we partition the stream of input rows  into {\em episodes} such that (1) an episode starts at time $t=1$ and also at time $t' \in [T]$ if row $x^{t'-1}$ and row $x^{t'}$ differ, and (2) an episode ends at the end of the stream and also at time $t'\in [T]$ if row $x^{t'}$ and row $x^{t'+1}$ differ. Our proof idea is to show that in no episode more than 9 intervals are closed. As there are at most $K_t+1$ episodes by the definition of episodes, Lemma~\ref{applem:switch} will follow. For episodes in which no interval is closed, nothing has to be shown. Thus, we study in the following episodes in which at least one interval is closed.
We first show the following claims, which basically show that if the mechanism receives the same row $x^t$  for a ``long enough'' time period, then eventually it will set the mode vector equal to $x^t$.

  Let $c_i^t$ (resp.~$\Thresh_{i,1}^t$ resp.~$\Thresh_{i,2}^t$)
  denote the value of $c_i$ (resp.~$\Thresh_{i,1}$ resp.~$\Thresh_{i,2}$) after the initial processing of time step $t$, i.e., in line~\ref{appcond:under_thresh_hist}.

\begin{claim}\label{appclaim:const}
Let $I$ be an episode in which at least one interval is closed and let $t^*$ be the first time step at which an interval is closed in $I$.
Conditioned on event $E$, if at any time step $t > t^*$ where $t \in I$ row $x^t$ equals the vector $\mode$ then no mode is updated in time step $t$.
\end{claim}
\begin{proof}
Consider a time step $t > t^*$ in $I$ and let $t'$ with $t^* \le t' < t$ be the last time before $t$ that a mode was changed. Then $c_i$ is reset to $0$ at time $t'$, for all $i\in[d]$.
By definition of $t'$, $\mode_i$ did not change since $t'$, and by definition of $I$, $x$ did not change since $t'$. Thus, $c^t = (t-t')\cdot x^t = t_{\diff} \cdot x^t = t_{\diff} \cdot \mode$.

Thus for each $i \in [d]$, $c_i^t - \Thresh_{i,1}^t = 2 \alpha_t$ and $c_i^t - \Thresh_{i,2}^t = -2 \alpha_t$ for all $i \in [d]$. It follows that  $\min_i (c_i^t - \Thresh_{i,1}^t) = 2\alpha_t$ and $\min_i (c_i^t - \Thresh_{i,2}^t) = -2\alpha_t$.
But since we condition on $E$, $|\mu_1^t|\le\frac{18 \log(4/\beta_t)}{\epsilon}$ and $|\tau_1| \le \frac{9\log(4/\beta_t)}{\epsilon}$, it follows that
$2\alpha_t = \frac{54}{\epsilon} \log (4/\beta_t) + \frac{6d}{\epsilon}\log (1/\beta_j) > \tau_1 - \mu_1^t$ and that $-2\alpha_t = -\frac{54}{\epsilon} \log (4/\beta_t) - \frac{6d}{\epsilon}\log (1/\beta_j) <\tau_2 - \mu_2^t$. Thus, the conditions on lines~\ref{appcond:under_thresh_hist} and ~\ref{appcond:over_thresh_hist} cannot hold and the interval  does not close at time $t$.
\end{proof}

\begin{claim}\label{appclaim:dec}
Let $I$ be an episode in which at least one interval is closed and let $t^*$ be the first time step at which an interval is closed in $I$.
Conditioned on event $E$,
at any time step $t> t^*$ where $t \in I$ and for any $i \in [d]$ if $\mode_i < x_i^t$ at the start of $t$ then $\mode_i$ will not decrease in time step $t$ and, symmetrically, if $\mode_i > x_i^t$ at the start of $t$ then $\mode_i$ will not increase in time step $t$.
\end{claim}
\begin{proof}
Consider a coordinate $i \in [d]$ and a time step $t > t^*$ in $I$ and let $t'$ with $t^* \le t' < t$ be the last time before $t$ that a mode was changed. Then $c_i$ is reset to $0$ at time $t'$, for all $i\in[d]$.
By definition of $t'$ $\mode_i$ did not change since $t'$, and by definition of $I$, $x_i$ did not change since $t'$. Thus, $c_i^t = (t-t')\cdot x_i^t = t_{\diff} \cdot x_i^t$.

If $x_i^t > \mode_i$, then
$c_i^t \ge  t_{\diff} \cdot \mode_i$. Due to the conditioning it follows that $c_i^t - \Thresh_{i,1}^t = c_i^t - t_{\diff} \cdot \mode_i + 2\alpha_t \ge 2\alpha_t = \frac{27}{\epsilon} \log (4/\beta_t) + \frac{3d}{\epsilon}\log (1/\beta_j) + \alpha_t >\alpha_t - \lambda_i$, and, thus, the condition in Line~\ref{appcond:update_mode_down} does not hold in time step $t$ and $\mode_i$ will not decrease.

If $x_i^t < \mode_i$, then
$c_i^t \le  t_{\diff} \cdot \mode_i$. Due to the conditioning it follows that $c_i^t - \Thresh_{i,2}^t = c_i^t  - t_{\diff} \cdot \mode_i - 2\alpha_t \le -2\alpha_t = -\frac{27}{\epsilon} \log (4/\beta_t) - \frac{3d}{\epsilon}\log (1/\beta_j) - \alpha_t < -\alpha_t - \lambda_i$, and, thus, the condition in Line~\ref{appcond:update_mode_up} does not hold in time step $t$ and $\mode_i$ will not increase.
\end{proof}
    The proof of Lemma~\ref{applem:switch} now consists of a careful case analysis using Claim~\ref{appclaim:const} and Claim~\ref{appclaim:dec} to show that within any episode, at most 9 intervals will be closed.
\begin{proof}[Proof of Lemma~\ref{applem:switch}]
    The claim that there are at most $t$ closed intervals follows trivially as at most one interval is closed in a single time step.

    We proceed to show that there are at most $9 K_t + 9$ closed intervals up to time step $t$.
    By Claim~\ref{appclaim:jointup} whenever an interval ends, at least one mode has to change. Thus, we will study in the following how many time steps  exist that contain a mode  update.
    By the definition of $K_t$ there are exactly $K_t + 1$ many episodes up to time step $t$ and, thus, it suffices to show that for any episode $I$, there are at most $9$ time steps in $I$ where a mode is updated.

    If no interval is closed in $I$, i.e., no mode is ever updated, the claim holds trivially for $I$. Thus in the following assume that there is at least one time step where a mode is updated and let $t^*$ be the first such time step.  As a shorthand we use $m_i$ to denote %
    the value of $\mode_i$ at the end of time step $t^*$.
    Note that  $c_i^{t^*} = 0$ and $x_i^t=x_i^{t^*}$ for all $i \in [d]$ and all $t\in I$.
    Thus for all subsequent time steps $t > t^*$ in $I$ and for all coordinates $i,j$ with $x_i^t = x_j^t$ it holds that $c_i^t = c_j^t$.
    Thus, by Claim~\ref{appclaim:jointup}, every time $t$ an interval is closed, there exists an $x^*\in\{-1,0,1\}$ and $m^*\in\{-1,0,1\}$ such that $\mode_i$ is updated for all $i$ with $x^t_i=x^*$ and $m_i=m^*$. As there are only 3 possible values that a variable $x^t_i=x^{t^*}_i$ can assume, it suffices to study for each value $x^* \in \{-1, 0, 1\}$  how often a coordinate $i$ with $x^{t^*}_i = x^*$ updates its mode within $I$.
  We analyze three cases:

    First we consider all coordinates $i$ with $x_i^{t^*} = -1$ and partition them into 3 subgroups depending on their $m_i$ value. For $m_i = -1$, Claim~\ref{appclaim:const} shows that there are no time steps with mode updates as the value of the mode and $x_i$ are equal. For $m_i = 0$, Claim~\ref{appclaim:dec} shows that there is at most one time step with mode updates, which decreases the mode to -1. For $m_i = 1$, Claim~\ref{appclaim:dec} shows that there are at most two time steps with mode updates, each decreasing the mode by 1. Thus there are at most 3 time steps with mode updates for all coordinates $i$ with $x_i^{t^*} = -1$.

    Next  consider all coordinates $i$ with $x_i^{t^*} = 0$ and partition them into 3 subgroups depending on $m_i$. For $m_i = 0$, there are no time steps with mode updates as the value of the mode and $x_i$ are equal. For $m_i = 1$, there is at most one time step with mode updates, which decreases the mode to 0. For $m_i = -1$, there is at most one time steps with mode updates,  which increases the mode to 0. Thus there are at most 2 time steps with mode updates for all coordinates $i$ with $x_i^{t^*} = 0$.

     Finally we consider all coordinates $i$ with $x_i^{t^*} = 1$ and partition them into 3 subgroups depending on $m_i$. For $m_i = 1$, Claim~\ref{appclaim:const} shows that there are no time steps with mode updates as the value of the mode and $x_i$ are equal. For $m_i = 0$, there is at most one time step with mode updates, which increases the mode to 1. For $m_i = -1$, there are at most two time steps with mode updates, each increasing the mode by 1. Thus there are at most 3 time steps with mode updates for all coordinates $i$ with $x_i^{t^*} = 1$.

     Thus, combined with the update at time step $t^*$ a mode update happens in at most 9 time steps in episode $I$.
     This concludes the proof.
    \end{proof}
Lemma~\ref{applem:switcherror} now follows by Lemma~\ref{appswitchacc:item} and Lemma~\ref{applem:switch}.
 
\section{\texorpdfstring{An $\Omega(d\cdot \log T)$ Lower Bound for Independently Differentially Private $d$-dimensional Binary Counting}{An Omega(d logT) Lower Bound for Independently Differentially Private d-dimensional Binary Counting}}\label{appsec:lowerbound}
So far, all upper bounds for \histogram\ and even for \maxsum\ which were not polynomial in $T$ relied on running $d$ independent binary counting mechanisms in parallel, and all achieved an error $\Omega(d \log T)$ for $\epsilon$-differential privacy. In this section we prove that using this strategy one cannot do better. For this, we formally define the following alternative version of differential privacy and neighboring streams of elements from $\{0,1\}^d$:

\begin{definition}[Independent differential privacy]
Let $x$ be a stream of $T$ elements from $\{0,1\}^d$. We say $x$ and $y$ are \emph{independently neighboring} if and only if for every $i\in [1,d]$ there exists a time step $t_i$ such that $x^t_i=y^t_i$ for all $t\neq t_i$. A mechanism $A$ is \emph{independently} $\epsilon$-\emph{differentially private} if it fulfills Definition \ref{def:dp} for independently neighboring $x$ and $y$.
\end{definition}

Note that this is a superset of the earlier definition of neighboring streams, i.e., all $x$ and $y$ which are neighboring are also independently neighboring, but not vice-versa. Thus, independent differential privacy is a stronger property then differential privacy.

We show the lower bound using a \emph{packing argument} \citep{Hardt2010}, which relies on the \emph{group privacy} property of differential privacy summarized in Fact \ref{appfact:group-privacy}.
We say $x$ and $y$ are $k$-neighboring if there exist $x=X_1, X_2, \dots, X_k=y$ such that $X_i$ and $X_{i+1}$ are neighboring for all $1\leq i <k$.
In the same way, we say $x$ and $y$ are independently $k$-neighboring if there exist $x=X_1, X_2, \dots, X_k=y$ such that $X_i$ and $X_{i+1}$ are independently neighboring for all $1\leq i <k$.
\begin{fact}\label{appfact:group-privacy} Let $A$ be an $\epsilon$-(independently) differentially private mechanism and $x$ and $y$ be (independently) $k$-neighboring. Then for all $S\in\mathrm{range}(A)$
    \begin{align*}
        P(A(x)\in S)\leq e^{k\epsilon}P(A(y) \in S)
    \end{align*}
    \end{fact}

Note that computing all $d$ column sums by running independent binary counting mechanisms fulfills independent $\epsilon$-differential privacy. Next we show that $\Omega(d\log T)$ noise is necessary for computing the noisy column sums in every time step while preserving independent $\epsilon$-differential privacy.

\begin{theorem}
Assume $d\leq \sqrt{T}$. Then there is a $T'=T'(\epsilon)$ such that any independently $\epsilon$-differentially private mechanism for computing all $d$ column sums cannot be $(\alpha,\beta)$-accurate for constant $\beta$ and $\alpha\leq \frac{d\ln T}{16\epsilon}$ for streams of length $T\geq T'$.
\end{theorem}
\begin{proof}
    Let $b=\frac{d\ln T}{8\epsilon}$.
    We assume without loss of generality that $T$ is a multiple of $b$, otherwise we pad the stream with zero vectors. We start by dividing $[1,T]$ into $T/b$ blocks of length $b$, i.e. $B_1=[1,b]$, $B_2=[b+1,2b]$, \dots, $B_{T/b}=[T-b+1,T]$.
    Now, for any vector $v\in [T/b]^d$ define the following stream $x(v)$ and output data set $S(v)$:
    \begin{itemize}
        \item $x(v)^t_i=1$ if and only if $t\in[(v_i-1)\cdot b +1, v_i \cdot b]$, else $x(v)^t_i=0$. That is, for every coordinate $d$ there is exactly one block which consists of only ones, and that block is the one specified by the $i$th coordinate in $v$.
        \item Denote $s_i^t$ the estimate for the $i$th column sum output by the mechanism in time $t$. $S(v)$ includes all outputs such that $s_i^t<b/2$ for all $t\leq (v_i-1)\cdot b$ and $s_i^t>b/2$ for all $t\geq v_i\cdot b$.
    \end{itemize}
Now let $\beta=1/3$ and assume there is a mechanism $A$ which is $(\alpha,\beta)$ accurate for $\alpha=\frac{d\ln T}{16\epsilon}=b/2$.
 Then $P(A(x(v))\in S(v))\geq \frac{2}{3}$.
    Further, we have that for any $v,v'\in [T/b]^d$, $x(v)$ and $x(v')$ are independently $2b$-neighboring: for every coordinate $i$, $x(v')_i$ and $x(v)_i$ differ in at most $2b$ timesteps.
 Therefore, by Fact \ref{appfact:group-privacy}, $P(A(x(v))\in S(v'))\geq \frac{2}{3}e^{-2b\epsilon}$.
 Additionally, all $S(v)$ are disjoint for distinct $v$.
Thus
    \begin{align*}
    1&\geq \sum_{v\in [T/b]^d} \frac{2}{3}e^{-2b\epsilon}\\
    &=\frac{2}{3}\left(\frac{8T\epsilon}{d\ln T}\right)^d e^{-\frac{d\ln T}{4}}\\
    &\geq \frac{2}{3}\left(\frac{8\sqrt{T}\epsilon}{\ln T}\right)^d(1/\sqrt[4]{T})^d\\
    &=\frac{2}{3}\left(\frac{8\sqrt[4]{T}\epsilon}{\ln T}\right)^d>1
    \end{align*}
    for large enough $T$, which is a contradiction.
\end{proof}
\section{Usage of Concurrent Composition}
\label{appsec:concurrent}

Our $\epsilon$-dp result could alternatively be shown as follows: One could use the result of \cite{QiuYi} to argue adaptive parallel composition for the partitioning mechanism, then Fact~\ref{fact:epsadaptive} by \cite{neurips2022} to argue that the $\epsilon$-dp partitioning mechanism in the continual release model is also $\epsilon$-dp in the adaptive continual release model, and then use the result of \citet{VW21concurrent} to concurrently compose the adaptive partitioning and adaptive histogram mechanisms. However, this would not reduce the technical complexity of the proof, and also not be self-contained.

Moreover, this proof strategy does not work for $(\epsilon, \delta)$-dp at all. Adaptive parallel composition in the $(\epsilon, \delta)$-dp setting is an open problem.
\citet{GMPS24neighboring} give an adaptive parallel composition theorem for $(\epsilon, \delta)$-dp, but their result assumes that the partition of the dataset is given beforehand, while we require that the partition of the dataset is also performed adaptively. Further, it is not enough to show that the partitioning mechanism is differentially private, we would need to show that it is \emph{adaptively} private since a general transformation as in the $\epsilon$-dp case does not exist here, and is only known for specific mechanisms. %
\section{\texorpdfstring{Extension to $(\epsilon, \delta)$-Differential Privacy}{Extension to (epsilon, delta)-Differential Privacy}}
\label{appsec:histqueryed}

We will use noise drawn from the Normal distribution for our mechanism. The mechanism constructed using noise drawn from the Normal distribution is known as the Gaussian mechanism, which satisfies $(\epsilon, \delta)$-dp.

\begin{definition}[Normal Distribution] The \emph{normal distribution} centered at $0$ with variance $\sigma^2$ is the distribution with the probability density function
\begin{align*}
f_{N(0,\sigma^2)}(x)=\frac{1}{\sigma\sqrt{2\pi}}\exp\left(-\frac{x^2}{2\sigma^2}\right)
\end{align*}

\end{definition}
 We use $X\sim N(0,\sigma^2)$ or sometimes just $N(0,\sigma^2)$ to denote a random variable $X$ distributed according to $f_{N(0,\sigma^2)}$.

\begin{fact}[Theorem A.1 in \cite{journals/fttcs/DworkR14}: Gaussian mechanism]\label{applem:gaussianmech}
 Let $f$ be any function $f:\chi\rightarrow \mathbb{R}^m$ with $L_2$-sensitivity $\Delta_2$.
Let $\epsilon\in(0,1)$, $c^2>2\ln(1.25/\delta)$, and $\sigma\geq c\Delta_2(f)/\epsilon$. Let $Y_i\sim N(0,\sigma^2)$ for $i\in[m]$. Then the mechanism defined as:
\begin{align*}
A(x)=f(x)+(Y_1,\dots,Y_m)
\end{align*}
satisfies $(\epsilon,\delta)$-differential privacy.
\end{fact}

We use the following continual histogram mechanism $H$ introduced by \cite{DBLP:conf/icml/FichtenbergerHU23}, which achieves an error of $O(\epsilon^{-1}\log(1/\delta)\log t\sqrt{d\ln(dt)})$ at time step $t$.
Since their mechanism fulfills the conditions of Theorem 2.1 of \cite{neurips2022}, the same privacy guarantees hold for their mechanism in the adaptive continual release model.

\begin{fact}[$(\epsilon,\delta)$-differentially private continual histogram against an adaptive adversary]
\label{appfact:epsdelthist}
    There is an $(\epsilon,\delta)$-differentially private mechanism in the adaptive continual release model for continual histogram that with probability $\ge 1-\beta$, has error bounded by $O(\epsilon^{-1}\log(1/\delta)\log t\sqrt{d\ln(dt/\beta)})$ at time $t$.
\end{fact}

\subsection{Histogram Queries}
We make the following changes to the mechanism to obtain an $(\epsilon, \delta)$-dp mechanism for histogram queries.

\begin{enumerate}
    \item Initialize an $\edee$-adaptively dp continual histogram mechanism $H$.
    \item Sample $\gamma_k^j \sim \edgammarv$.
    \item Set $\agammaj$ to $\edu$.
\end{enumerate}

\paragraph{Privacy.}
We detail the changes to the privacy proof from the $\epsilon$-dp case.
As in the $\epsilon$-dp case, we need to now show that
\[
\Pr\left[ \alg{x} \in S \right]
\le e^{\epsilon} \cdot \Pr\left[ \alg{y} \in S \right] + \delta
\]
Since $H$ is $\edee$-adaptively differentially private, we get that
\begin{align*}
\Pr(V_{H,Adv(x,y)}^{(x)}\in S)\leq e^{\he}\Pr(V_{H,Adv(x,y)}^{(y)}\in S) + \hd
\end{align*}
and
\begin{align*}
\Pr(V_{H,Adv(x,y)}^{(y)}\in S)\leq e^{\he}\Pr(V_{H,Adv(x,y)}^{(x)}\in S) + \hd.
\end{align*}
Thus all we would need to show would be
\begin{align}\label{appeq:edviewsxyswitchedK}
\Pr(V_{H,Adv(x,y)}^{(x)}\in S)\leq e^{2\epsilon/3}\Pr(V_{H,Adv(y,x)}^{(x)}\in S) + \delta/2,
\end{align}
since then
\begin{align}\begin{split}\label{appeq:edfullprivacyK}
\Pr(\mathcal{A}(x)\in S)&=\Pr(V_{H,Adv(x,y)}^{(x)}\in S)\\
&\leq e^{2\epsilon/3}\Pr(V_{H,Adv(y,x)}^{(x)}\in S) + \delta/2\\
&\leq e^{\epsilon}\Pr(V_{H,Adv(y,x)}^{(y)}\in S) + \delta\\
&=e^{\epsilon}\Pr(\mathcal{A}(y)\in S) + \delta
\end{split}
\end{align}
The partitioning is still $e^{\epsilon/3}$-close by the same arguments since we use the same random variables as in the $\epsilon$-dp case.
For the thresholds, note that conditioned on all previous outputs of $H$ and $p_j$ being equal, $q_k(s^{p_j}(x))$ and $q_k(s^{p_j}(y))$ can differ by at most $1$ for each $k\in[m]$.
Thus the $L_2$ difference between the two vectors is at most $\sqrt{m}$. By Lemma~\ref{applem:gaussianmech} for the Gaussian mechanism, adding $\edgammarv$ noise to every ${q_k}(s^{p_j}(y))$ ensures that the distributions of ${q_k}(s^{p_j}(x)) + \gamma_k^j$ and ${q_k}(s^{p_j}(y)) + \gamma_k^j$ are $(e^{\epsilon/3}, \delta/2e^{2\epsilon/3})$-close for all $k\in[m]$. Since the condition in line \ref{appline:threshold} only depends on those, this implies that the probabilities of executing line \ref{appline:threshold} on any subset of $[m]$ on $\mathrm{run}(x)$ and $\mathrm{run}(y)$ are $(e^{\epsilon/3}, \delta/2e^{\epsilon/3})$-close, as required.

\paragraph{Accuracy.}

We have that \bounds\ holds with $\amut, \atauj$ as earlier, $\agammaj=\edu$ and $\ahj = O(\epsilon^{-1}\log(1/\delta)\log j\sqrt{d\ln(dj/\beta)})$. Thus, by Lemma~\ref{applem:acccases}, the mechanism has error at most
\[
\alphat = O \left( \epsilon^{-1}  \log(1/\delta) \cdot \left( \sqrt{d} \log^{3/2} (dm{\cmax}^t/\beta) + \sqrt{m} \log (m{\cmax}^t/\beta)  + \log t \right)\right)
\]
at all time steps $t$ with probability at least $1 - \beta$
as required.

\subsection{Histogram Parameterized in the Number of Fluctuations}
We make the following changes to the mechanism to obtain an $(\epsilon, \delta)$-dp mechanism for \histogram\ parameterized in the number of fluctuations.

\begin{enumerate}
    \item Initialize an $(\eps/3, \delta/2)$-adaptively dp continual histogram mechanism $H$.
    \item Sample $\gamma_i^j \sim N(0, 18d \ln (4 e^{2\eps/3}/\delta) / \eps^2$.
    \item Replace $\frac{3d}{\eps} \log(1/\beta_j)$ with $6 \eps^{-1} \sqrt{d \ln (4 e^{2 \eps/3}/\delta \beta_j)}$.
\end{enumerate}

\paragraph{Privacy.}
We detail the changes to the privacy proof from the $\epsilon$-dp case.

We will show that the computation of $p_0, p_1,\ldots$ along with the sequence of $\mode$ value updates is $(2\eps/3, \delta/2)$-dp. Then the result follows by basic composition with the histogram mechanism.

The partitioning is still $e^{\epsilon/3}$-close by the same arguments since we use the same random variables as in the $\epsilon$-dp case.
For the modes, note that conditioning on ending the $j$-th interval at $p_j$, $\|c^{p_j}(x) - c^{p_j}(y)\|_2 \le \sqrt{d}$. Thus, by Lemma~\ref{applem:gaussianmech}, adding $N(0, 18d \ln (4 e^{2\eps/3}/\delta) / \eps^2$ noise to each $c^{p_j}_i$ ensures that the distributions on $x$ and $y$ are $(e^{\epsilon/3}, \delta/2e^{\epsilon/3})$-close. Thus both the partitioning and mode updates together are $(2\eps/3, \delta/2)$-dp as required.

\paragraph{Accuracy.}

Using the histogram mechanism from Fact~\ref{appfact:epsdelthist}, and replacing $\frac{d}{\eps} \log(1/\beta_j)$ in the accuracy proofs with $6 \eps^{-1} \sqrt{d \ln (4 e^{2 \eps/3}/\delta \beta_j)}$, we get that the mechanism has error at most
\[
O\left( \eps^{-1} \log(1/\delta) \cdot \left(
\sqrt{d} \log^{3/2}(dK/\beta) + \log t
\right)\right)
\]
at all time steps $t$ with probability at least $1 - \beta$.
\section{The Sparse Vector Technique}
\label{appsec:sparsevector}

The sparse vector technique is based on an algorithm by \cite{DBLP:conf/stoc/DworkNRRV09} and was described more fully by \cite{journals/fttcs/DworkR14}. The version described in Algorithm \ref{appalg:sparsevector} is by \cite{journals/pvldb/LyuSL17} for $c=1$ (the main difference is that it allows different thresholds for every query).

\begin{algorithm}[!htbp]
\begin{algorithmic}[1]
\State{\textbf{Input: }}{Data Set $D$, Sensitivity bound $\Delta$, thresholds $\T_1,\T_2,\dots$, and queries $q_1,q_2,\dots$ which are have sensitivity at most $\Delta$}

\State $\tau=\Lap(2\Delta/\epsilon)$

\For{$i=1,\dots,$}
    \State $\mu_i=\Lap(4\Delta/\epsilon)$
    \If{$q_i(D)+\mu_i>\T_i+\tau$}
        \State output $a_i=\yes$
        \State {\bf Abort}
    \Else
        \State output $a_i=\no$
    \EndIf
\EndFor
\caption{AboveThreshold}
\label{appalg:sparsevector}
\end{algorithmic}
\end{algorithm}

\begin{lemma}[\citealp{journals/fttcs/DworkR14, journals/pvldb/LyuSL17}]\label{applem:SVpriv}
    Algorithm \ref{appalg:sparsevector} is $\epsilon$-differentially private.
\end{lemma}

\begin{lemma}[\citealp{journals/fttcs/DworkR14, journals/pvldb/LyuSL17}]\label{applem:SVacc}
    Algorithm \ref{appalg:sparsevector} fulfills the following accuracy guarantees for $\alpha=\frac{8(\ln k + \ln(2/\beta))}{\epsilon}$:
    For any sequence $q_1,\dots,q_k$ of queries it holds with probability at least $1-\beta$,
    \begin{enumerate}
        \item for $i$ such that $a_i=\yes$ we have
        \begin{align*}
            q_i(D)\geq \T_i-\alpha,
        \end{align*}
        \item for all $i$ such that $a_i=\no$ we have
        \begin{align*}
            q_i(D)\leq \T_i+\alpha.
        \end{align*}
    \end{enumerate}
\end{lemma}

\end{document}